%% file: main.tex
\documentclass{article}
\usepackage[preprint,nonatbib]{neurips_2024}

\usepackage{placeins}
\usepackage[utf8]{inputenc}
\input{macros}
\usepackage{array}
\usepackage{hyperref}

\usepackage{enumitem}

\title{Boosted Conformal Prediction Intervals}
\author{%
  Ran Xie \\
  Department of Statistics\\
  Stanford University\\
  \texttt{ranxie@stanford.edu} \\
   \And
   Rina Foygel Barber\\
    Department of Statistics\\
    University of Chicago \\
   \texttt{rina@uchicago.edu} \\
    \AND
    Emmanuel J. Candès\\
    Department of Statistics\\
Department of Mathematics\\
Stanford University\\
   \texttt{candes@stanford.edu} \\
}

\usepackage[
backend=biber,
style=numeric,
sorting=ynt
]{biblatex}
\begin{document}
\maketitle
\begin{abstract}
This paper introduces a \textit{boosted conformal procedure} designed to tailor conformalized prediction intervals toward specific desired properties, such as enhanced conditional coverage or reduced interval length. We employ machine learning techniques, notably gradient boosting, to systematically improve upon a predefined conformity score function. This process is guided by carefully constructed loss functions that measure the deviation of prediction intervals from the targeted properties. The procedure operates post-training, relying solely on model predictions and without modifying the trained model (e.g., the deep network).  
{Systematic experiments demonstrate that starting from conventional conformal methods, our boosted procedure achieves substantial improvements in reducing interval length and decreasing deviation from target conditional coverage.}

\end{abstract}
\section{Introduction}
Black-box machine learning algorithms have been increasingly employed to inform decision-making in sensitive applications. For instance, deep convolutional neural networks have been applied to diagnose skin cancer \cite{Esteva2017}, and AlphaFold has been utilized in the development of malaria vaccines \cite{alphafold, Higgins2022}; here, scientists have employed AlphaFold to predict the structure of a key protein in the malaria parasite, facilitating the identification of potential binding sites for antibodies that could prevent the transmission of the parasite \cite{Higgins2022}. These instances highlight the critical need for understanding prediction accuracy, and one popular approach to quantify the uncertainty associated with general predictions relies on the construction of prediction sets guaranteed to contain the target label or response with high probability. Ideally, we would like the coverage to be valid conditional on the values taken by the features of the predictive model (e.g., patient demographics). 

Conformal prediction \cite{Vovk2005} stands out as a flexible calibration procedure that provides a wrapper around any black-box prediction model to produce valid prediction intervals. Imagine we have a data set $\{(X_i,Y_i)\}_{i=1}^n$ and a test point $(X_{n+1},Y_{n+1})$ drawn exchangeably from an unknown, arbitrary distribution $P$ (e.g. the pairs $(X_i,Y_i)$ may be i.i.d.). Taking the data set and the observed features $X_{n+1}$ as inputs, conformal prediction forms a prediction interval ${C_n}(X_{n+1})$ for $Y_{n+1}$ with valid marginal coverage, i.e.~such that $\mathbbm{P}(Y_{n+1}\in {C_n}(X_{n+1})) = 0.95$ or any nominal level specified by the user ahead of time. {This is achieved by means of a conformity score $E(x,y;f)$, where $(x,y)$ represents a data point while $f$ represents any aspects of the distribution that we have estimated. For instance, the score may be given by the magnitude of the prediction error $|y-\hat\mu(x)|$, where $\hat\mu(x)$ represents the model prediction of the expected outcome, in which case $f$ is simply $\hat\mu$. Roughly,} we would include $y$ in the prediction interval if $E(X_{n+1},y  {;f})$ does not take on an atypical value when compared with $\{E(X_i,Y_i  {;f})\}$, $i =1, \ldots, n$.  Selecting an appropriate conformity score is akin to choosing a test statistic in statistical testing, where two statistics may yield the same Type I error rate yet differ substantially in other aspects of performance. 

One central issue is that while the conformal procedure guarantees marginal coverage, it does not extend similar guarantees to other desirable inferential properties without additional assumptions. In response, researchers have introduced a variety of conformity scores, including the locally adaptive (Local) conformity score \cite{lei2018}, the conformalized quantile regression (CQR) conformity score \cite{romano2019}, and its variants, CQR-m \cite{kivaranovic2020adaptive} and CQR-r \cite{Sesia_2020}. Among these, CQR has often demonstrated superior empirical performance in terms of both interval length and conditional coverage \cite{romano2019}.

This paper introduces a boosting procedure aimed at enhancing an arbitrary score function.\footnote{An implementation of the boosted conformal procedure (BoostedCP) is available online at \url{https://github.com/ran-xie/boosted-conformal}.} By employing machine learning techniques, namely, gradient boosting, our objective is to modify the Local or CQR score functions (or other baselines) to reduce the average length of prediction intervals or improve conditional coverage while maintaining marginal coverage. While this paper focuses primarily on length and conditional coverage, our methods can be tuned to optimize other criteria; we elaborate on this in Section \ref{sec:discussion}.

Our boosted conformal procedure searches within a family of generalized scores for a score achieving a low value of a loss function adapted to the task at hand. Specifically, to evaluate the conditional coverage of prediction intervals, we build a loss function that maximizes deviation from the target coverage rate in the leaves of a shallow contrast tree \cite{Friedman2020}. Searching within a strategically designed family of score functions, rather than directly retraining or fine-tuning the fitted model under the task-specific loss function, yields greater flexibility and avoids the costs associated with retraining or fine-tuning. Further, this boosting process is executed post-model training, requiring only the model predictions and no direct access to the training algorithm. 

{Source code for implementing the boosted conformal procedure is available online at \url{https://github.com/ran-xie/boosted-conformal}. Details regarding the acquisition and preprocessing of the real datasets are also provided in the GitHub repository.}

\begin{figure}[t]
    \centering
    \includegraphics[width=\linewidth,trim={0cm 0.5cm 0cm 0cm},clip]{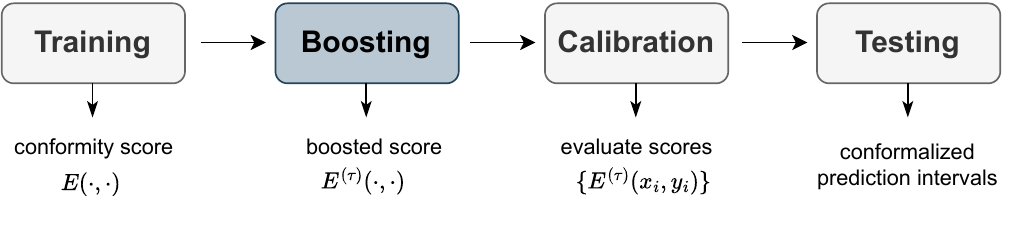}
    \caption{Illustration of the boosted conformal prediction procedure. We introduce a boosting stage between training and calibration, where we boost $\tau$ rounds on the conformity score function $E(\cdot,\cdot)$ and obtain the boosted score $E^{(\tau)}(\cdot,\cdot)$. The number of boosting rounds $\tau$ is selected via cross validation.
    A detailed description of the procedure is presented in Algorithm \ref{alg:cap}.
    }
    \label{fig:diagram}
\end{figure}
\section{
The split conformal procedure}\label{subsec:prelim}
We begin by outlining the key steps of the split conformal procedure applied to a family  $\{(X_i,Y_i)\}_{i=1}^n$ of exchangeable samples (e.g., i.i.d.).
\begin{itemize}[label={\scriptsize$\bullet$}]
    \item {\em Training.} Randomly partition $[n]$ into a training set $I_1$ and a calibration set $I_2$. On the training set, train a model by means of an algorithm $\mathbbm{A}$ to produce a conformity score function $E(\cdot,\cdot;f)$. The structure of this score function is predetermined, whereas the model $f$ is learned from $\mathbbm{A}$. An example of a conformity score is $E(x,y;f)=|y-\hat{\mu}(x)|$, where $\hat{\mu}(x)$ is a learned regression function so that $f$ is here simply $\hat{\mu}$.
    \item {\em Calibration.} Evaluate the function $E(\cdot,\cdot;f)$ on each instance in the calibration set and obtain scores $\{E_i\}_{i\in I_2}$,\footnote{The term `score' will henceforth refer to the conformity score unless stated otherwise.} with each $E_i=E(X_i,Y_i;f)$. The $(1-\alpha)$th empirical quantile of the score, $Q_{1-\alpha}(E,I_2)$, is calculated as
    \begin{align*}
        Q_{1-\alpha}(E,I_2)=\inf\{z: \mathbbm{P}\left(Z\leq z\right)\geq 1-\alpha\},
    \end{align*}
    where $Z$ follows the distribution $\frac{1}{|I_2|+1}(\delta_{\infty}+\sum \delta_{E_i})$, and $\delta_a$ is a point mass at $a$.
    \item {\em Testing.} For a new observation $X_{n+1}$, output the conformalized prediction interval 
    \begin{align}\label{eq: general.conf.int}
        {C_n}(X_{n+1})=\{y\in\mathbbm{R}: E(X_{n+1},y;f)\leq Q_{1-\alpha}(E,I_2)\}.
    \end{align}
\end{itemize}
If {ties between $\{E_i\}_{i\in I_2}$ occur with probability zero}, it holds that 
\begin{align}\label{eq:conformal}
     1- \alpha \le \mathbbm{P}(Y_{n+1}\in {C_n}(X_{n+1})) \le 1-\alpha+\frac{2}{|I_2|+2},
\end{align}
see \cite{lei2018}. 
By introducing additional randomization during the calibration step, the prediction interval can be tuned to obey $\mathbbm{P}(Y_{n+1}\in  {C_n}(X_{n+1}))=1-\alpha$, see \cite{vovk2005algorithmic}.  This adjustment is not critical here and we omit the details.

\textbf{Locally adaptive conformal prediction} (Local for short) \cite{lei2018} introduces a score function that aims to make conformal prediction adapt to situations where the spread of the distribution of $Y$ varies significantly with the observed features $X$. On the training set, run an algorithm $\mathbbm{A}$ to fit two functions $\mu_0(\cdot)$ and $\sigma_0(\cdot)$, where $\mu_0(X)$ estimates the conditional mean $\mathbb{E}[Y \mid X]$, and $\sigma_0(X)$ the dispersion around the conditional mean, frequently chosen as the conditional mean absolute deviation (MAD), $\mathbb{E} [|Y-\mu_0(X)| \mid X]$. With $f=(\mu_0,\sigma_0)$, the locally adaptive (Local) score function is:
\begin{align}
E(x,y;f) = |y - \mu_0(x)|/\sigma_0(x).\label{local.conf.score}
\end{align}
For a new observation $X_{n+1}$, the conformalized prediction interval  (\ref{eq: general.conf.int}) takes on the simplified expression 
    $[{\mu}_0(X_{n+1})-{Q}_{1-\alpha}(E,I_2){\sigma}_0(X_{n+1}), {\mu}_0(X_{n+1})+{Q}_{1-\alpha}(E,I_2){\sigma}_0(X_{n+1})]$.

\textbf{Conformalized quantile regression} (CQR) \cite{romano2019conformalized} also aims to adapt to heteroskedasticity by calibrating conditional quantiles, which often results in  shorter prediction intervals.  Apply quantile regression to produce a pair of estimated quantiles $(\hat{q}_{\alpha/2}(x),\hat{q}_{1-\alpha/2}(x))$, where $\hat{q}_\beta(X)$ is the estimated $\beta$th quantile of the conditional distribution of $Y$. The CQR score function is defined as
\begin{align}\label{score.cqr}
    E(x,y;f)= \max\{\hat{q}_{\alpha/2}(x)-y,y-\hat{q}_{1-\alpha/2}(x)\},
\end{align}
where $f=(\hat{q}_{\alpha/2},\hat{q}_{1-\alpha/2})$.  
For a new observation $X_{n+1}$, following (\ref{eq: general.conf.int}) yields the prediction interval
\begin{align}
    \left[\hat{q}_{\alpha/2}(X_{n+1})-{Q}_{1-\alpha}(E,I_2),\hat{q}_{1-\alpha/2}(X_{n+1})+{Q}_{1-\alpha}(E,I_2)\right].\label{eq:conf.cqr}
\end{align}

\textbf{Generalized conformity score families}. To construct a Local conformity score, we estimate two functions $\mu_0(\cdot)$ and $\sigma_0(\cdot)$ to plug into \eqref{local.conf.score}. Since these components are constructed without looking at performance downstream, it is reasonable to imagine that other choices may enjoy enhanced properties. How then should we systematically select $\mu(\cdot)$ and $\sigma(\cdot)$? To address this, we define a generalized Local score family $\mathcal{F}$ containing all potential score functions of the form
\begin{align}\label{eq:local.gen}
    \mathcal{F}:=\{E(\cdot,\cdot  {;f}):E(x,y; f)=|y-\mu{(x)}|/\sigma(x),\sigma(\cdot) >0\},
\end{align}
  {where $f=(\mu,\sigma)$.} For each $E(\cdot,\cdot  {;f}) \in \mathcal{F}$, the conformalized prediction interval is given by
\begin{align}\label{eq:conf.local.b}
    \left[{\mu}(X)-{Q}_{1-\alpha}(E,I_2){\sigma}(X), {\mu}(X)+{Q}_{1-\alpha}(E,I_2){\sigma}(X)\right].
\end{align}

Turning to CQR, one notable limitation is the uniform adjustment of prediction intervals by the constant factor ${Q}_{1-\alpha}(E,I_2)$, as shown in (\ref{eq:conf.cqr}). This approach is suboptimal in the presence of heteroskedasticity, as it applies an identical correction to prediction intervals of varying widths for each $X=x$. 
{Thus, simply updating the fitted quantiles $(\hat{q}_{\alpha},\hat{q}_{1-\alpha/2})$ and plugging them into the original score function would be inadequate, as the structure of the original score imposes significant limitations on the effectiveness of conformalized prediction intervals.} To address this, several variants including CQR-m \cite{kivaranovic2020adaptive} and CQR-r \cite{Sesia_2020} have been proposed.   {Focusing on CQR-r, it employs a flexible score function, defined as $
    E(x,y;f)= \max\{\hat{q}_{\alpha/2}(x)-y,y-\hat{q}_{1-\alpha/2}(x)\}/(\hat{q}_{1-\alpha/2}(x)-\hat{q}_{\alpha/2}(x))$, with $f=(\hat{q}_{\alpha/2},\hat{q}_{1-\alpha/2},\hat{q}_{1-\alpha/2}-\hat{q}_{\alpha/2})$.}
Following (\ref{eq: general.conf.int}), conformalized prediction intervals become
\begin{align}
    \left[\hat{q}_{\alpha/2}(X)-\hat{\sigma}(X){Q}_{1-\alpha}(E,I_2),\hat{q}_{1-\alpha/2}(X)+\hat{\sigma}(X){Q}_{1-\alpha}(E,I_2)\right],
\end{align}
where $\hat{\sigma}=\hat{q}_{1-\alpha/2}-\hat{q}_{\alpha/2}$. Intuitively, the adjusted score function allows prediction bands to adjust in proportion to their width, instead of adding a constant shift as in CQR. However, despite the intuitive appeal of adjusted scores as a seemingly more reasonable ``allocation'' of the conformal correction, empirical studies reveal that they do not result in narrower prediction intervals when compared to CQR \cite{Sesia_2020}. This phenomenon is largely due to the uniform direction of the conformal adjustment, represented by ${Q}_{1-\alpha}(E,I_2)$, across all observations. In particular, if ${Q}_{1-\alpha}(E,I_2)<0$, indicating that the true target $y$ predominantly lies within the estimated quantile range $[\hat{q}_{\alpha/2},\hat{q}_{1-\alpha/2}]$, there is a uniform narrowing of the predicted interval across all samples.

In light of these insights, we propose a novel score family, $\mathcal{H}$, designed to augment the flexibility of the conformity score functions:
\begin{align}\label{eq:cqr.gen}
\hspace*{-0.6em}\mathcal{H}:=\bigl\{E(\cdot,\cdot  {;f}):E(x,y  {;f})=\max\left\{{\mu_1}(x)-y,y-\mu_2(x)\right\}/{\sigma}(x),\mu_1(\cdot)\leq \mu_2(\cdot),\sigma(\cdot)> 0\bigr\},
\end{align}
  {where $f=(\mu_1,\mu_2,\sigma)$}, which leads to conformalized prediction intervals of the form 
\begin{align}
 [{\mu_1}(X)&-{\sigma}(X){Q}_{1-\alpha}(E,I_2),{\mu_2}(X)+{\sigma}(X){Q}_{1-\alpha}(E,I_2)].\label{eq:conf.cqr.gen}
\end{align}
{Notably, $\mathcal{H}$ includes the Local, CQR, and CQR-r scores as special cases.}

\section{Boosted conformal procedure}\label{subsec:boosted.cp}

It is clear from above that a model is trained to produce a conformity score $E(\cdot,\cdot;f)$; e.g., we may learn a regression function $\hat{\mu}(\cdot)$ to plug it into a score function $|y-\hat{\mu}(x)|$. 
To overcome the limitation of working with an arbitrarily selected score function, we introduce a boosting step before calibration, see Figure \ref{fig:diagram}. In a nutshell, we use gradient boosting to iteratively improve upon a predefined score $E(\cdot,\cdot;f)$ now denoted as $E^{(0)}(\cdot,\cdot)$, where the superscript indicates the $0$th iteration.  

To achieve this, we construct a task-specific loss function $\ell$, which takes a dataset $\mathcal{D}$ and a score function $E(\cdot,\cdot  {;f})$ as inputs, and outputs $\ell(E(\cdot,\cdot  {;f});\mathcal{D})$ measuring how closely the conformalized prediction interval aligns with the analyst's objective. This loss function $\ell$ is designed to be differentiable with respect to each of the model components produced by the training algorithm. Importantly, it does not require knowledge of the gradient of $f(x)$ with respect to $x$. In the example above, taking the labels as fixed, this means that for each feature $x_i \in \mathcal{D}$, $i=1,\ldots,n$,   {if we set $\hat{y}_i=\hat{\mu}(x_i)$, then the loss $\ell(E(\cdot,\cdot);\mathcal{D})$ is a function of $\{\hat{y}_i\}_{i=1}^n$, and the derivative $\partial\ell(E(\cdot,\cdot);\mathcal{D})/\partial\hat{y}_i$ is well defined.}  
 {In Sections \ref{subsec:eval.cond} and \ref{subsec:len.der}, we present examples of such derivatives.}

 Each boosting iteration updates the score function sequentially, employing a gradient boosting algorithm such as XGBoost \cite{chen2016xgboost} or LightGBM \cite{ke2017lightgbm}. These algorithms accept as input a dataset $\mathcal{D}$, a base score function $E(\cdot,\cdot;f)$, a custom loss function $\ell$, gradients of $\ell$ with respect to $f$ (denoted $\nabla_f \ell$), and a number of boosting rounds $\tau$. We may write the boosting procedure as
\begin{align}\label{eq:gb.step}
    (E^{(0)}(\cdot,\cdot),\ldots,E^{(\tau)}(\cdot,\cdot))=\text{GradientBoosting}(\mathcal{D},E(\cdot,\cdot;f),\ell,\nabla_f\ell,\tau).
\end{align}
This yields a boosted score function $E^{(\tau)}(\cdot,\cdot)$, which is then used for calibration and for constructing prediction intervals. The number $\tau$ is calculated using $k$-fold cross-validation on the training dataset, selecting $\tau$ from potential values up to a predefined maximum $T$ (e.g., 500).  
We partition the dataset into $k$ folds and for each $j = 1, \ldots, k$,  we hold out fold $j$ for sub-calibration and the remaining $k-1$ folds for sub-training. We apply $T$ rounds of gradient boosting (\ref{eq:gb.step}) on the sub-training data, generating $T+1$ candidate score functions   {$E_j^{(0)}(\cdot,\cdot),\ldots, E_j^{(T)}(\cdot,\cdot)$}. 
Each score function is then evaluated on sub-calibration data, using the loss function $\ell$ to compute losses at all epochs, i.e., for each fold $j=1,\ldots k$,
\begin{align*}
    \{L_j^{(t)}\}_{t=0}^T= \{\ell(E_j^{(t)};\text{fold}_j)\}_{t=0}^T.
\end{align*}
Last, $\tau$ is selected as the round that minimizes the average loss across all $k$ folds:
\begin{align}\label{eq:cv.step}
    \tau=\arg{\textstyle\min}_{0\leq t\leq T} {\textstyle\sum}_{j=1}^kL_j^{(t)},
\end{align}
see Figure \ref{fig:diagram_cv}. This cross-validation step simulates the calibration step in conformal prediction and effectively prevents the overfitting of the score function. 

Since boosting is conducted on the training data, the boosted procedure satisfies {the same} marginal coverage guarantee as the split conformal procedure, as formalized below.
\begin{prop}\label{prop:cov}
    Let $\{(X_i,Y_i)\}_{i=1}^m$ be the held out calibration set, and $(X_{m+1},Y_{m+1})$ be a pair of new observation. If the $m+1$ samples are exchangeable, and ties between $\{E^{(\tau)}(X_i,Y_i)\}_{i=1}^m$ occur with probability zero, the confromalized prediction interval (\ref{eq: general.conf.int}) computed from score function $E^{(\tau)}(\cdot,\cdot)$ satisfies 
{the} coverage guarantee (\ref{eq:conformal}).
\end{prop}

\begin{figure}[t]
    \centering
    \includegraphics[width=\linewidth]{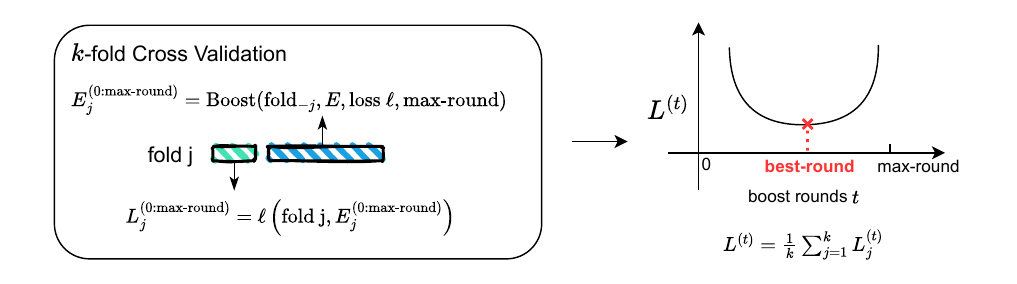}
    \vspace{-0.5cm}
    \caption{Schematic drawing showing the selection of the number of boosting rounds via cross-validation. Left: we hold out fold $j$, and use the remaining $k-1$ folds to generate candidate scores $E_j^{(t)}$, $t = 0, \ldots, \text{max-round}$. 
    The performance of each score is evaluated on fold $j$ using the loss function $\ell$. Right: $\text{best-round}$ minimizes the average loss across all $k$ folds.
    A detailed description of the procedure is presented in Algorithm \ref{alg:cap}.}
    \label{fig:diagram_cv}
\end{figure}

\textbf{Searching within generalized conformity score families}.
To update the Local score function (\ref{local.conf.score}), we search within the generalized score family $\mathcal{F}$ (\ref{eq:local.gen}). First, we initialize $\mu^{(0)}=\mu_0$ and $\sigma^{(0)}=\sigma_0$. After completing $\tau$ iterations of boosting on the training set, we obtain the boosted score function $E^{(\tau)}(x,y)=|y - \mu^{(\tau)}(x)|/\sigma^{(\tau)}(x)$. Notably, we can update any score function within $\mathcal{F}$. For instance, to update  $E(x,y;f)=|y-\hat{\mu}(x)|$, we simply initialize $\mu^{(0)}=\hat{\mu}$, and take $\sigma^{(0)}$ to be the constant function equal to one. Similarly, to update the CQR score function (\ref{score.cqr}), we search within the score family $\mathcal{H}$ (\ref{eq:cqr.gen}). First, we initialize a triple $\mu_1^{(0)}=\hat{q}_{\alpha/2}$, $\mu_2^{(0)}=\hat{q}_{1-\alpha/2}$, $\sigma^{(0)}=\hat{q}_{1-\alpha/2}-\hat{q}_{\alpha/2}$. After $\tau$ boosting rounds, we obtain the boosted score function $E^{(\tau)}(x,y) = \max\{{\mu_1^{(\tau)}}(x)-y,y-{\mu_2}^{(\tau)}(x)\}/{\sigma}^{(\tau)}(x)$.

\begin{algorithm}[ht]
\caption{Boosting stage}\label{alg:cap}
\textbf{Input:}
\begin{algorithmic} 
\State $\text{Training data }(X_i,Y_i)\in\mathbbm{R}^p\times \mathbbm{R}$, $i=1,...,n$; base conformity score function $E^{(0)}(\cdot,\cdot)$
\State Loss function $\ell$; target mis-coverage level $\alpha\in (0,1)$
\State Number $k$ of cross-validation folds; maximum boosting rounds $T$ 
\end{algorithmic}
\textbf{Procedure:}
\begin{algorithmic} 
\State Randomly divide $\{1,...,n\}$ into $k$ folds
\For{\texttt{$j \gets 1$ to $k$}}
        \State Set fold $j$ as sub-calibration set, and the remaining $k-1$ folds as sub-training set
        \State On the sub-training set, call GradientBoosting (\ref{eq:gb.step}) to obtain candidate scores $\{E_j^{(t)}\}_{t=0}^T$
    \State On the sub-calibration set, evaluate $L_j^{(t)}=\ell(E_j^{(t)})$, $t=0,\ldots,T$
      \EndFor
    \State Set boosting rounds $\tau\gets \arg\min_t \frac{1}{k}\sum_{j=1}^kL_j^{(t)}$ as in (\ref{eq:cv.step})
    \State On the training set, call GradientBoosting (\ref{eq:gb.step}) to obtain boosted functions $\{E^{(t)}\}_{t=0}^{\tau}$
\end{algorithmic}
\textbf{Output:}
\begin{algorithmic}
 \State Boosted conformity score function $E^{(\tau)}(\cdot,\cdot)$
\end{algorithmic}
\end{algorithm}

\section{Related Works}
\textbf{Adapting the classical conformal procedure} to improve properties of the conformalized intervals has been one of the primary focuses of recent literature. Noteworthy contributions---including CF-GNN \cite{huang2023uncertainty} and ConTr \cite{stutz2022learning}---approach this problem by introducing modifications to the training stage of the procedure. As outlined in Section \ref{subsec:prelim}, a model is trained to produce a score function $E(\cdot,\cdot;f)$. The model $f$ usually depends on a set of model parameters, e.g., neural network parameters $\theta$. Denote the trained model $f$ by $f_{\theta}$. CF-GNN and ConTr retrain or fine-tune the model by using a carefully constructed loss function, which may aim to produce narrower
prediction intervals or prediction sets of reduced cardinality in classification problems. This process generates a new set of model parameters $\theta'$. The new model $f_{\theta'}$ is then plugged into the {\em same} predefined conformity score function---namely CQR \cite{huang2023uncertainty} or the adaptive prediction set score (APS) \cite{stutz2022learning}---to produce $E(\cdot, \cdot, f_{\theta'})$. 

There are two primary limitations. First, the score function imposes constraints on the properties of conformalized intervals as explained in Section \ref{subsec:prelim}. Our approach introduces more flexibility by constructing a family of generalized score functions that is a superset of 
$\{E(\cdot,\cdot;f_{\theta}):\theta\in\Theta\}$, where $\Theta$ is the parameter space of the training model. 
This family is  strategically designed to contain an 
oracle conformity score ideally suited to the task at hand, e.g., achieving exact conditional coverage. Second, current methodologies necessitate fine-tuning or retraining models from scratch, requiring both access to the training model and significant computational resources. In contrast, our boosted conformal method operates directly on model predictions and circumvents these issues.

\textbf{Conditional coverage} of conformalized prediction intervals has also attracted significant interest, characterized by efforts to establish theoretical guarantees and achieve numerical improvements. Prior work established an impossibility result \cite{vovk2012conditional,barber2020}, which states that exact conditional coverage in finite samples cannot be guaranteed without making assumptions about the data distribution. Subsequently, Gibbs et al. \cite{gibbs2023conformal} developed a modified conformal procedure that guarantees conditional coverage for predefined protected sub-groups, i.e.~subsets of the feature space. 
Our approach differs from the previous works by introducing a numerical method directly aimed at improving the conditional coverage, $\mathbbm{P}(Y\in {C_n(}X)|X=x)$, across all potential values of $x$. 

\section{Boosting for conditional coverage}\label{sec:boost.cond.cov}
Maintaining valid marginal coverage, our goal is to produce a prediction interval ${C_n}$ obeying 
\begin{align}\label{eq:cond.cov.condition}
    \mathbbm{P}(Y\in {C_n}(X_{n+1})|X_{n+1}=x) \approx 1-\alpha
\end{align}
for all possible values of $x$. To this end, we present a loss function that quantifies the conditional coverage rate of any prediction interval. Requiring merely a dataset $\mathcal{D}$ and a prediction interval ${C_n(\cdot)}$ as inputs, it also serves as an effective evaluation metric, which may be of independent interest.

\subsection{A measure for deviation from target conditional coverage}\label{subsec:eval.cond}
{From now on, we let $E$ be the score function $E(\cdot,\cdot;f)$.} 
Set $\mathcal{D}=\{(X_i,Y_i)\}_{i=1}^n$ and denote by ${C_n(\cdot)}$ the conformalized prediction interval constructed from ${E}$. We shall assess the deviation of ${C_n(\cdot)}$ from the target conditional coverage by means of Contrast Trees \cite{Friedman2020}. As background, a contrast tree 
iteratively identifies splits within the feature space $\mathcal{X}$ in a greedy fashion, aiming to maximize absolute within-group deviations from the target conditional coverage rate $(1-\alpha)$. For a subset $R$ of the data point indices $[n]$, let $\mathcal{D}_R=\{X_j,Y_j\}_{j\in {R}}$. The absolute within-group deviation is computed as
\begin{align}\label{eq:lack.condcov.measure.i}
d\left({C_n(\cdot)};\mathcal{D}_{{R}}\right)=\left| {|{R}|}^{-1}{\textstyle\sum}_{j\in {R}}\mathbbm{1}(Y_j\in {C_n(}X_j))- (1-\alpha) \right|.
\end{align}
The overall empirical maximum deviation is then defined as
\begin{align}\label{eq:lack.condcov.measure}
\ell_M\left({E};\mathcal{D}\right) = {\textstyle\max}_{1\leq m\leq M}{d\left({C_n(\cdot)};\mathcal{D}_{\hat{R}_m}\right)},
\end{align}
where  $\hat{R}_1 \cup \dots \cup \hat{R}_M$ is a partition of $[n]$, which itself depends on $E$ and $\mathcal{D}$. Specifically, it is computed by running a contrast tree for $M$ iterations. At each iteration, the algorithm not only seeks to isolate regions with large deviations but also discourages {splits where any subset $\hat{R}_m$ is too small}.

To update score functions via gradient boosting as described in (\ref{eq:gb.step}), we would need a differentiable approximation of the maximum deviation. {To this end, we construct approximations for the following three components of the loss function. With an abuse of notation, in subsequent discussions, we shall employ the same notations to denote these differentiable approximations. }
\begin{enumerate}
    \item Approximation for the prediction interval ${C_n(\cdot)}$ in (\ref{eq:lack.condcov.measure.i}): the prediction interval is formulated as (\ref{eq:conf.local.b}) for the generalized Local score, and as (\ref{eq:conf.cqr.gen}) for the generalized CQR score. Denote the upper and lower limits of ${C_n(\cdot)}$ by $u(\cdot)$ and $l(\cdot)$. 
    We approximate the empirical quantile $Q_{1-\alpha}(E,I_2)$ in $u(\cdot)$ and $l(\cdot)$ with a smooth quantile estimator {$Q^s_{1-\alpha}$}. 
Given $r$ scalars $\{z_i\}_{i=1}^r$, ${Q^s_{1-\alpha}}$ is constructed as:
\begin{align}\label{eq:qs}
    {Q^s_{1-\alpha}}(\{z_i\}_{i=1}^r):= \langle\text{HD}(r),s(\mathbf{z})\rangle,
\end{align}
where $\langle\cdot,\cdot\rangle$ represents the dot product. Here, $\text{HD}(r)=[W_{r,1},...,W_{r,r}]$ is the weight vector corresponding to the Harrel-Davis distribution-free empirical quantile estimator \cite{HarrellDavis}, and $s(\mathbf{z})$ is a differentiable ordering $\{\Tilde{z}_{(i)}\}_{i=1}^r$, arranged in the ascending order. {In practice,  the derivative of $s(\textbf{z})$ with respect to each $z_i$ is given by the package developed in \cite{pmlr-v119-blondel20a}.} This approach is a smooth approximation of the Harrel-Davis quantile estimator $Q^{\text{HD}}_{1-\alpha}$, 
constructed as a linear combination of the order statistics, $Q^{\text{HD}}_{1-\alpha}=\langle\text{HD}(r),\{z_{(i)}\})\rangle={\textstyle\sum}_{i=1}^rW_{r,i}z_{(i)}$, where $W_{r,i}$ takes the value $I_{(1-\alpha)(r+1),\alpha(r+1)}(i/r)-I_{(1-\alpha)(r+1),\alpha(r+1)}((i-1)/r)$ and $I_{a,b}(x)$ represents the incomplete beta function. 
    \item Approximation for absolute deviation $d_i$ (\ref{eq:lack.condcov.measure.i}): the indicator function  in  (\ref{eq:lack.condcov.measure.i}) can be approximated by the product of two sigmoid functions,
    \begin{align*}
        \mathbbm{1}(Y_j\in {C_n}(X_j))&= \mathbbm{1}({u}(X_j)-Y_j\geq 0) \mathbbm{1}(Y_j-{l}(X_j)\geq 0)\\
        &\approx S_{\tau_1}({u}(X_j)-Y_j)S_{\tau_1}(Y_j-{l}(X_j)),
    \end{align*}
    where $\tau_1$ is a parameter, trading off smoothness and quality of the approximation. The sigmoid function $S_{\tau_1}(x)$ is defined as $        S_{\tau_1}(x) = (1+e^{-\tau_1 x})^{-1}$.

    \item Approximation for maximum deviation: we employ a log-sum-exp  function \cite{boyd2004convex} 
    to derive the differentiable approximation of $\ell_M$ as
    \begin{align}
      {\ell}_M\left({E};\mathcal{D}\right):={\tau_2}^{-1}\log{\textstyle\sum}_{m=1}^M\exp{\left(\tau_2 {d}_m({C_n}(\cdot);\mathcal{D}_m)\right)},\label{eq:proxy}
    \end{align}
    where $\tau_2$ is a parameter, serving the same purpose as $\tau_1$.
\end{enumerate}

{Here, we demonstrate calculating the derivative of the smooth approximation (\ref{eq:proxy}) with respect to each component of the generalized Local score, expanding it as follows:}
\begin{align*}
    &{\ell}_M\left({E};\mathcal{D}\right)={\tau_2}^{-1}\log\sum_{m=1}^M\exp{\left(\tau_2 \left| {|R_m|}^{-1}{\textstyle\sum}_{j\in R_m} S_{\tau_1}({u}(X_j)-Y_j)S_{\tau_1}(Y_j-{l}(X_j))- (1-\alpha) \right|\right)},
\end{align*}
where 
\begin{align*}
    S_{\tau_1}({u}(X_j)-Y_j)=\left(1+\exp\left[-\tau_1 (\mu_j+{Q}_{1-\alpha}^s(\{E_i\}_{i=1}^n)\sigma_j-Y_j)\right]\right)^{-1},\\
    S_{\tau_1}(Y_j-{l}(X_j))=\left(1+\exp\left[-\tau_1 (Y_j-\mu_j+{Q}_{1-\alpha}^s(\{E_i\}_{i=1}^n)\sigma_j)\right]\right)^{-1},
\end{align*}
with $\mu_i=\mu(X_i)$, $\sigma_i=\sigma(X_i)$, $E_i = |Y_i-\mu_i|/\sigma_i$. As a result,  for each feature $X_i$ within $\mathcal{D}$, we can evaluate $\partial{\ell}_M\left({E};\mathcal{D}\right)/\partial\mu_i$ and $\partial{\ell}_M\left({E};\mathcal{D}\right)/\partial\sigma_i$ via the chain rule. 

\subsection{Boosting score functions for conditional coverage}\label{sec:condcov.semi}
Since the empirical maximum deviation $\ell_M$ (\ref{eq:lack.condcov.measure}) is non-differentiable, we opt for the differentiable approximation during the gradient boosting step (\ref{eq:gb.step}). Nonetheless, we utilize the original $\ell_M$ to select the number of boosting rounds as in step (\ref{eq:cv.step}) and to evaluate the conditional coverage of the conformalized prediction interval on the test set. 
 
\subsubsection{Theoretical guarantees}
The oracle score function achieving conditional coverage as defined in (\ref{eq:cond.cov.condition}) belongs to both proposed generalized score families. 

\begin{prop}[Asymptotic expressiveness]\label{thm:uni.condcov}
    Let $\{X_i,Y_i\}_{i=1}^n$ be i.i.d. with continuous joint probability density distribution. Under the split conformal procedure, for any target coverage rate $1-\alpha$, as $n\rightarrow \infty$, there exists $(\mu^*,\sigma^*)$ and $(\mu_1^*,\mu_2^*,\sigma^*)$ such that the corresponding generalized Local (\ref{eq:local.gen}) and CQR  (\ref{eq:cqr.gen}) score functions recover conditional coverage at rate $1-\alpha$, as defined in (\ref{eq:cond.cov.condition}).
\end{prop}
It goes without saying that there is no reason to assume that the optimal  $\mu^*$ corresponds to the conditional mean, median or any quantile of $Y$ given $X$, or that the optimal $\sigma^*$ corresponds to the standard deviation or the mean absolute deviation of $Y$ given $X$, as in the original Local score (\ref{local.conf.score}). That said, our greedy strategy has no guarantee on global optimality and this is why the choice of the starting point---whether it is the Local or CQR score function---plays a role in the performance.

\subsubsection{Empirical results on real data}\label{sec:emp.res}
We apply our boosted conformal procedure to the 11 datasets previously analyzed in \cite{Sesia_2020,romano2019,kivaranovic2020adaptive}. Details on the datasets are provided in Section \ref{sec.app: data} in the Appendix. In each dataset, we randomly hold out $20\%$ as test data. All experiments are repeated 10 times, starting from the data splitting. We refer to Section \ref{sec:model.hyperparams} for details on the models and hyper-parameters we employ for the training and boosting stages. 

We evaluate the conditional coverage of the prediction intervals as the maximum within-group deviations across a partitioned test set (\ref{eq:lack.condcov.measure}). This partition is obtained through a contrast tree algorithm described in Section \ref{subsec:eval.cond}. Figure \ref{fig:condcov.local.meps19} illustrates the comparison between miscoverage rates of prediction intervals at each leaf of the contrast tree. These intervals are derived under the classical Local conformal procedure and our boosted conformal procedure. Notably, the conditional coverage of the boosted prediction interval more closely aligns with the target rate $1-\alpha$.

 The experiment results summarized in Table \ref{tab:condcov} indicate that applying boosting significantly enhances the performance of the baseline Local procedure. In contrast, boosting on CQR does not yield significant improvements---a sign that CQR already targets conditional coverage. (Before boosting, the prediction intervals generated by the baseline Local procedure exhibit conditional coverage deviations up to three times greater than those of the baseline CQR procedure.)  It is noteworthy, however, that after boosting, the conditional coverage of the Local procedure improves to a level comparable to that of the boosted CQR procedure. While generally slightly less effective, nevertheless surpasses the performance of the boosted CQR procedure in two cases. Results on the remaining datasets are deferred to Tables \ref{tab:condcov.local.full} and \ref{tab:condcov.cqr.full}.

\begin{figure}
\centering
 \includegraphics[width=\linewidth]{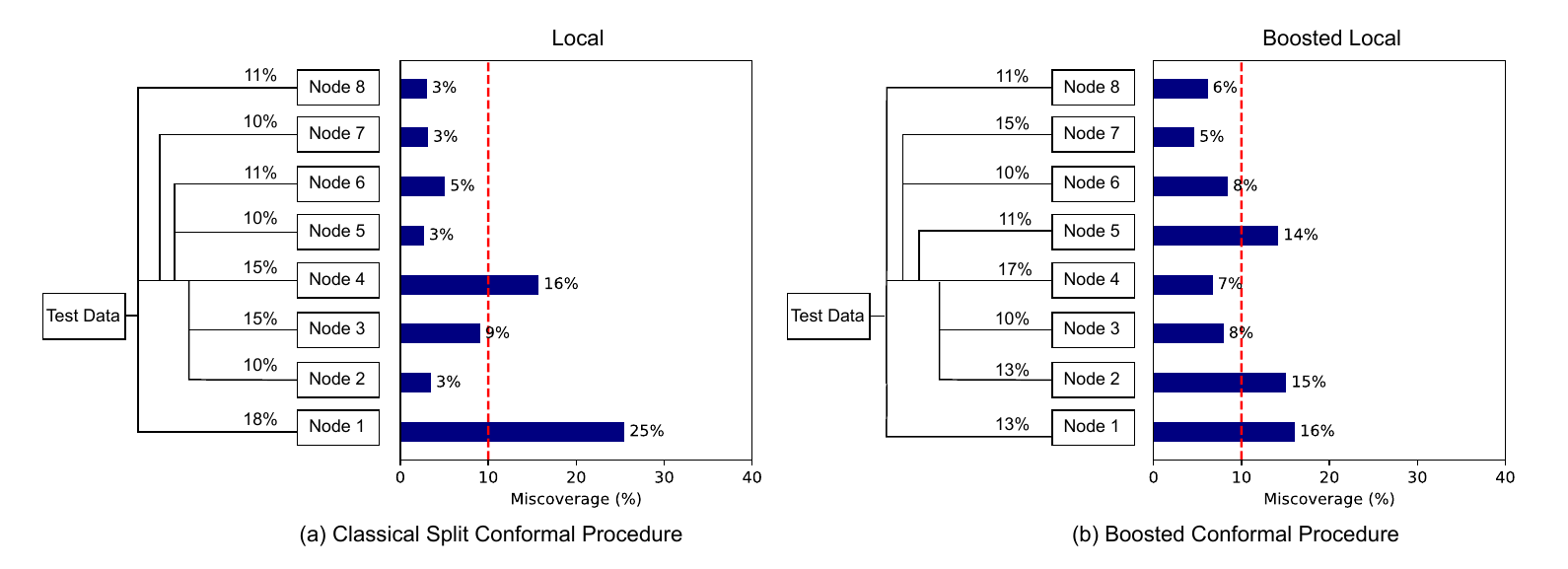}
 \vspace{-0.5cm}
    \caption{Comparison of test set conditional coverage evaluated on the dataset {meps-19}: (a) shows the classical Local-type conformal procedure and (b) our boosted Local-type conformal procedure.  The target miscoverage rate is set to $\alpha=10\%$ (red). Miscoverage rate is computed at each leaf of the contrast tree, constructed to detect deviation from the target rate. 
    Each leaf node is labeled with its size, namely, the fraction of the test set it represents.}
    \label{fig:condcov.local.meps19}
\end{figure}
 \begin{table}[tb]
  \begin{center}
  \vspace{-0.4cm}
    \caption{Test set maximum deviation loss $\ell_M$ evaluated on various conformalized intervals. The best result achieved for each dataset is highlighted in bold. }
    \vspace{0.3cm}
    \label{tab:condcov}
    {\footnotesize
    \begin{tabular}{S S S S S S S}
    \hline
      \multicolumn{7}{c}{Max. Conditional Coverage Deviation ($\%$), target miscoverage $\alpha = 10\%$}\\
      \hline
      \multicolumn{1}{c}{\multirow{2}{*}{Dataset}}& \multicolumn{2}{c}{Method} &  \multicolumn{1}{c}{\multirow{2}{*}{Improvement}}& \multicolumn{2}{c}{Method} &  \multicolumn{1}{c}{\multirow{2}{*}{Improvement }}\\ 

      \multicolumn{1}{c}{} & \multicolumn{1}{c}{Local} & $\text{Boosted}$&\multicolumn{1}{c}{} & \multicolumn{1}{c}{CQR} & $\text{Boosted}$&\multicolumn{1}{c}{}\\ 
      \hline
      $\text{bike}$& 10.979432&\text{5.638} & $\text{-48.65}\%$&4.933807&\textbf{4.925}&\text{-0.17}\%\\
      bio&5.303131 & \text{4.862}& $\text{-8.31}\%$&5.068905&\textbf{4.700}&\text{-7.29}\%\\
      community&25.754902 & \text{13.466}& $\text{-47.71}\%$&12.687841&\textbf{12.105}&\text{-4.59}\%\\
      $\text{concrete}$& 10.739634&\text{8.763} & $\text{-18.40}\%$&9.038694&\textbf{8.265}&\text{-8.56}\%\\
      $\text{meps-19}$ & 15.357278& \text{5.656}& $\text{-63.17}\%$&5.506775&\textbf{5.507}&\text{-0.00}\%\\
      $\text{meps-20}$ & 16.938720 & \textbf{6.998}& $\text{-58.69}\%$&7.614116&\text{7.184}&\text{-5.65}\%\\
      $\text{meps-21}$&17.627239 &\textbf{7.832} & $\text{-55.57}\%$&8.164735&\text{8.067}&\text{-1.20}\%\\
   
      \hline
    \end{tabular}}
  \end{center}
\end{table}

\section{Boosting for length}\label{sec:boost.length}
We begin by specifying the oracle prediction interval with minimum length. For a random variable $Z$, the High Density Region (HDR) at a specified significance level $\alpha$, denoted as $\text{HDR}_{\alpha}(Z)$,  is defined as the shortest deterministic interval that covers $Z$ with probability at least $1-\alpha$. The boundaries of $\text{HDR}_{\alpha}(Z)$, the lower limit $Q_{l(\alpha)}$ and the upper limit $Q_{u(\alpha)}$, obey the condition $
    \mathbbm{P}(Z\in [Q_{l(\alpha)},Q_{u(\alpha)}])\geq 1-\alpha$.
For a pair of $(X,Y)$ drawn from $P$, for every value of $x \in \mathbb{R}^p$, the oracle prediction 

interval at that point is expressed as
\begin{align}\label{eq: oracle.len}
    \text{HDR}_{\alpha}(Y|X=x)=\left[Q_{l(\alpha)}(Y|X=x),Q_{u(\alpha)}(Y|X=x)\right].
\end{align}
Before introducing the boosting strategy, we present a word of caution against optimizing exclusively for this objective. Importantly, to maintain valid marginal coverage, the {shortest} prediction interval is prone to overcover when the spread of $Y|X$ (the conditional distribution of $Y$ given $X$) is small, and undercover when the spread of $Y|X$ is large. This may be undesirable.

Similar to {Proposition} \ref{thm:uni.condcov}, we can show that the generalized score families 

exhibit the necessary expressiveness to contain the oracle conformity score, achieving optimal length while ensuring valid marginal coverage. The formal proof is deferred to Section \ref{sec:thm.power}.

\subsection{A measure for length}\label{subsec:len.der}

Consider a dataset $\mathcal{D}=\{(X_i,Y_i)\}_{i=1}^n$ and a score function ${E}$. Denote the corresponding conformalized prediction interval by ${C_n(\cdot)}$, with its quality measured by the average length:
\begin{align}\label{eq:power.loss}
    \ell_L({E};\mathcal{D})=n^{-1}{\textstyle\sum}_{i=1}^n|{C_n(}X_i)|.
\end{align}

To derive a differentiable approximation of {$\ell_L$}, we approximate the empirical quantile $Q_{1-\alpha}$ in the conformalized intervals (\ref{eq:conf.local.b}) and (\ref{eq:conf.cqr.gen}) with the smooth quantile estimator ${Q^s_{1-\alpha}}$ constructed in (\ref{eq:qs}).   {Here, we demonstrate calculating the derivative of the smooth approximation of $\ell_L$ with respect to each component of the generalized Local score, expanding it as follows based on the previously outlined approximation steps:}
\begin{align*}
    \ell_L\left({E};\mathcal{D}\right)=2\Bar{\sigma}{Q}_{1-\alpha}^s(\{E_i\}_{i=1}^n), \qquad E_i = |Y_i-\mu_i|/\sigma_i, 
\end{align*}
with $\mu_i=\mu(X_i)$, $\sigma_i=\sigma(X_i)$, $\Bar{\sigma}=n^{-1}\sum_{i=1}^n\sigma_i$. As a result,  for each feature $X_i$ within $\mathcal{D}$, we can evaluate $\partial\ell_L\left({E};\mathcal{D}\right)/\partial\mu_i$ and $\partial\ell_L\left({E};\mathcal{D}\right)/\partial\sigma_i$ via the chain rule.  For instance,
\begin{align*}
\frac{\partial\ell_L\left({E};\mathcal{D}\right)}{\partial\mu_i}=-2\Bar{\sigma}\frac{\partial{Q}_{1-\alpha}^s(\{E_j\}_{j=1}^n)}{\partial E_i}\frac{\text{sign}(Y_i-\mu_i)}{\sigma_i}.
\end{align*}

\subsection{Empirical results on real data}
We apply our boosted conformal procedure to the same datasets described in Section \ref{sec:emp.res}. Detailed information on the models and hyperparameters used during the training and boosting stages can be found in Section \ref{sec:model.hyperparams}. 
Partial experiment results are summarized in Table \ref{tab:table1}.
Notably, the boosting performance highlighted in bold exhibits significant improvement compared to previously documented results \cite{romano2019conformalized,Sesia_2020}.  We see a pronounced enhancement with the blog dataset; before boosting, the Local prediction intervals are on average $42\%$ longer than those generated by CQR. After boosting, these intervals outperform the boosted CQR intervals by $32\%$. Using CQR as the baseline also yields substantial improvements, a decrease in averaged length exceeding $10\%$ in six out of the eleven datasets. The meps-21 dataset, in particular, shows an improvement of up to $18\%$ relative to the baseline. Results on the remaining datasets can be found in Tables \ref{tab:len.local.full} and \ref{tab:len.cqr.full}. Figure \ref{fig:len.real} compares the conformalized prediction intervals derived from baseline Local and CQR  scores with those obtained from the boosted scores. To effectively visualize the impact of boosting, we conduct a regression tree analysis on the training set {to predict the label $Y$}, setting the maximum number of tree nodes to four. This regression tree is then applied to the test set, allowing for a detailed comparison of the prediction intervals across each of the four distinct leaves. 

\begin{figure}[t]
    \centering
    \includegraphics[trim={1cm 0.5cm 0.5cm 1cm},clip,width=0.8\linewidth]{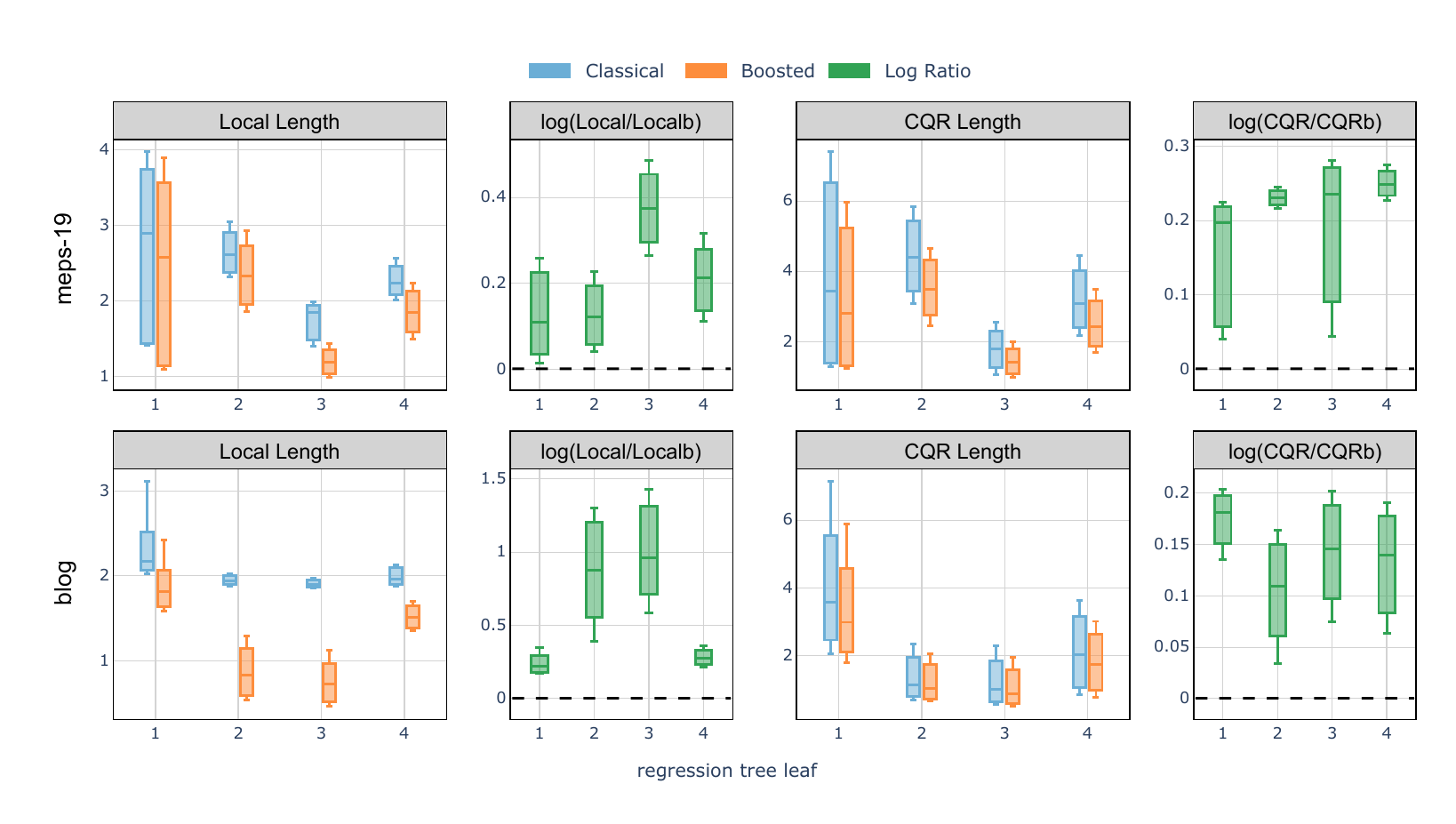}
     \caption{
    Comparison of test set average interval length {evaluated on the meps-19 and blog datasets}: classical Local and CQR conformal procedure versus the boosted procedures (abbreviated as `Localb' and `CQRb') compared in each of the 4 leaves of a regression tree trained on the training set {to predict the label $Y$}. A positive log ratio value between the regular and boosted interval lengths indicates improvement from boosting. The target miscoverage rate is set at $\alpha=10\%$. 
    }\label{fig:len.real}
    \vspace{-0.4cm}
\end{figure}
\begin{table}[htb]
  \begin{center}
    \caption{Test set {average interval length} $\ell_L$ evaluated on various conformalized prediction intervals. The best result achieved for each dataset is highlighted in bold. 
    }
    \label{tab:table1}
    \vspace{0.3cm}
    {\footnotesize
    \begin{tabular}{S S S S S S S}
    \hline
      \multicolumn{7}{c}{Average Length, target miscoverage $\alpha = 10\%$}\\
      \hline
      \multicolumn{1}{c }{\multirow{2}{*}{Dataset}}& \multicolumn{2}{c}{Method} &  \multicolumn{1}{c}{\multirow{2}{*}{Improvement}}& \multicolumn{2}{c}{Method} &  \multicolumn{1}{c}{\multirow{2}{*}{Improvement}}\\ 
      \multicolumn{1}{c}{} & \multicolumn{1}{c}{Local} & $\text{Boosted}$&\multicolumn{1}{c}{} & \multicolumn{1}{c}{CQR} & $\text{Boosted}$&\multicolumn{1}{c}{}\\ 
      \hline
        blog & 2.056256 & \textbf{0.972} & $\text{-52.74}\%$&1.444662 &{1.434}&\text{-0.71}\%\\
        $\text{facebook-1}$ & 1.896465 & \text{1.383}& $\text{-27.03}\%$&1.198184&\textbf{1.072}&{-10.47}\%\\
      $\text{facebook-2}$ & 1.854089 & \text{1.363} & $\text{-26.51}\%$&1.199786&\textbf{1.075}&{-10.41}\%\\
      $\text{meps-19}$ & 2.070403 & \textbf{1.685}& $\text{-18.60}\%$&2.553576& \text{2.136}& $\text{-16.35}\%$\\
      $\text{meps-20}$ & 2.081153 & \textbf{1.836}& $\text{-11.80}\%$&2.667130& \text{2.357}& $\text{-11.62}\%$\\
      $\text{meps-21}$ &2.062927  &\textbf{1.795}  & $\text{-12.99}\%$&2.584983& \text{2.105}& $\text{-18.55}\%$\\
            
      \hline
    \end{tabular}}
  \end{center}
  \vspace{-0.45cm}
\end{table}

\section{Discussion}\label{sec:discussion}
We introduced a post-training conformity score boosting scheme aiming to optimize for conditional coverage or {length} of the conformalized prediction interval. An intriguing avenue for future exploration involves simultaneously optimizing both {length} and conditional coverage, {potentially trading off these objectives by incorporating user-specified weights \cite{dabah2024temperature}}.  {Additionally, we can readily adapt our procedure to meet various application-specific objectives. For instance, we can optimize for conditional coverage on predefined feature groups, a common task in enhancing fairness in distributing social resources across different demographic groups \cite{gibbs2023conformal}. Similarly, we can modify our procedure to reduce the length of prediction intervals for predefined label groups, which can be seen as reallocating resources to decrease uncertainty for certain groups at the expense of higher uncertainty for other groups \cite{stutz2022learning}.} Candidate loss functions tailored to these objectives are detailed in Section \ref{sec.app: loss.gen}. Lastly, the primary emphasis of this paper centers on the design of the conformity score boosting scheme and formalizing the optimization of conditional coverage in mathematical terms, leaving room for computational optimization to enhance performance and runtime efficiency. {In essence, the gradient boosting algorithm in our procedure can be replaced with any gradient-based machine learning model. Thus, another interesting future direction would be to explore whether alternative algorithms could enhance performance.}
 
\vspace{-0.1cm}
\section*{Acknowlegements}
E.C. was supported by the Office of Naval Research under Grant No. N00014-24-1-2305, the National Science Foundation under Grant No. DMS2032014, and the Simons Foundation under Award 814641. R.F.B. was supported by the Office of Naval Research via grant N00014-20-1-2337, and by the National
Science Foundation via grant DMS-2023109.

\clearpage

\printbibliography

\newpage

\input{appendix_camera_ready}

\end{document}

%% file: macros.tex
\usepackage{subcaption}

\usepackage{graphicx}
\usepackage{comment}
\usepackage{xcolor}
\usepackage{float}
\usepackage{mathtools,amsmath,amssymb,amsthm,accents}
\usepackage{bbm}
\usepackage{enumitem}
\usepackage{algorithm}
\usepackage{algpseudocode}

\usepackage{booktabs}
\usepackage{multirow} %
\usepackage{siunitx} 
\newtheorem{theorem}{Theorem}[section]

\newtheorem{lemma}[theorem]{Lemma}
\newtheorem{prop}[theorem]{Proposition}

\newtheorem{definition}[theorem]{Definition}

\usepackage[normalem]{ulem}

\usepackage{scalerel,stackengine}
\stackMath
\newcommand\reallywidehat[1]{%
\savestack{\tmpbox}{\stretchto{%
  \scaleto{%
    \scalerel*[\widthof{\ensuremath{#1}}]{\kern-.6pt\bigwedge\kern-.6pt}%
    {\rule[-\textheight/2]{1ex}{\textheight}}
  }{\textheight}%
}{0.5ex}}%
\stackon[1pt]{#1}{\tmpbox}%
}

%% file: appendix_camera_ready.tex
\appendix
\counterwithin{figure}{section}
\counterwithin{table}{section}
\renewcommand\thefigure{\thesection\arabic{figure}}
\renewcommand\thetable{\thesection\arabic{table}}
\section{Appendix}
\subsection{Candidate loss functions for additional application-specific objectives}\label{sec.app: loss.gen}

\textbf{Conditional coverage on predefined feature groups}: this task can be viewed as a specialized application within our broader strategy of boosting for conditional coverage, as detailed in Section \ref{sec:boost.cond.cov}. There, the primary challenge was to develop a loss function that accurately measures deviations from the target conditional coverage rate. We achieved this by using contrast trees to identify partitions in the feature space that maximize these deviations, effectively identifying subgroups in need of protection. This process is simplified when the partitions correspond to prespecified groups, allowing us to continue using the empirical maximum deviation as a candidate loss function.

Consider a dataset $\mathcal{D}=\{(X_i,Y_i)\}_{i=1}^n$ and a score function ${E}$. Denote by ${C_n(}\cdot)$ the conformalized prediction interval constructed from ${E}$. Let $G_1,\ldots,G_M$ be prespecified feature index groups. Within each set $\mathcal{D}_i=\{(X_j,Y_j)\}_{j\in G_i}$, compute the absolute deviation $d_i$ as
\begin{align}
d_i\left({C_n(}\cdot);\mathcal{D}_i\right)=\left| \frac{1}{|G_i|}{\textstyle\sum}_{j\in G_i}\mathbbm{1}(Y_j\in {C_n(}X_j))- (1-\alpha) \right|.
\end{align}
The overall empirical maximum deviation is then defined as
\begin{align}
\ell\left({E};\mathcal{D}\right) = {\textstyle\max}_{1\leq i\leq M}{d_i\left({C_n(}\cdot);\mathcal{D}_i\right)}.
\end{align}

{\textbf{Interval length conditional on predefined label groups}: for a dataset $\mathcal{D}=\{(X_i,Y_i)\}_{i=1}^n$ and a score function ${E}$, let $\mathcal{Y}_1,\ldots,\mathcal{Y}_M$ be the prespecified label groups. A natural minimization objective for balancing uncertainty among these groups is defined as:
\begin{align*}
    \ell({E};\mathcal{D})=\sum_{i=1}^Mw_i\frac{1}{\sum_{j=1}^n\mathbbm{1}(Y_j\in \mathcal{Y}_i)}\sum_{j=1}^n\mathbbm{1}(Y_j\in \mathcal{Y}_i)|{C_n(}X_j)|,
\end{align*}
where $(w_1,\ldots,w_M)$ represents a set of user-specified weights.}

\subsection{Proof of {Proposition} \ref{thm:uni.condcov}}\label{sec:pf.thm}
Our proof relies on the following lemma.
\begin{lemma}[Expressiveness]\label{lem:uni}
Given any sample pair $X$ and $Y$ with a continuous joint probability density distribution, and a prediction interval $[c_l(\cdot),c_u(\cdot)]$ with marginal coverage equal to $1-\alpha$, there exist specific function sets: ($\mu(\cdot)$,$\sigma(\cdot)$) for the Local type, and ($\mu_1(\cdot)$,$\mu_2(\cdot)$,$\Tilde{\sigma}(\cdot)$) for the CQR type, such that asymptotically:
\begin{enumerate}
\item The conformalized prediction interval (\ref{eq:conf.local.b}), derived using the generalized Local type conformity score $f_{\mu,\sigma}$, accurately recovers $[c_l(\cdot),c_u(\cdot)]$.
\item Similarly, the conformalized prediction interval (\ref{eq:conf.cqr.gen}), based on the generalized CQR type conformity score $E_{\mu_1,\mu_2,\Tilde{\sigma}}$, also recovers $[c_l(\cdot),c_u(\cdot)]$.
\end{enumerate}
\end{lemma}
\begin{proof}[Proof of Lemma \ref{lem:uni}]
Recall that the generalized Local score (\ref{eq:local.gen}) characterized by $(\mu,\sigma)$ takes the form
\begin{align}
    E_{\mu,\sigma}(x,y)=\frac{|y-\mu{(x)}|}{\sigma(x)}.
\end{align}
Asymptotically, the conformalized prediction interval is given by
\begin{align}
    \left[{\mu}(X)-{Q}_{1-\alpha}(E_{\mu,\sigma}){\sigma}(X), {\mu}(X)+{Q}_{1-\alpha}(E_{\mu,\sigma}){\sigma}(X)\right].
\end{align}
Here, ${Q}_{1-\alpha}$ represents the population quantile.
Set
\begin{align*}
    \mu(x)&=\frac{c_{l}(x)+c_{u}(x)}{2},\quad
    \sigma(x)=\frac{c_{u}(x)-c_{l}(x)}{2}.
\end{align*}
By assumption, we have
\begin{align*}
    \mathbbm{P}(Y\in [c_{l}(X),c_{u}(X)])= 1-\alpha.
\end{align*}
With a simple change of variables, the above inequality is equivalent to
\begin{align*}
     \mathbbm{P}\left(\left|\frac{Y-\mu(X)}{\sigma(X)}\right|\leq 1\right)= 1-\alpha.
\end{align*}
In other words, this is equivalent to
\begin{align*}
    Q_{1-\alpha}(E_{\mu,\sigma})= 1.
\end{align*}
We have thus proved the result for the generalized Local type conformity
score $E_{\mu,\sigma}$. In the same spirit, we can prove the result for the generalized CQR type conformity score $E_{\mu_1,\mu_2,\Tilde{\sigma}}$ by taking 
\begin{align*}
    \mu_1=c_l+\frac{c_{u}(x)-c_{l}(x)}{2},\quad \mu_2=c_u-\frac{c_{u}(x)-c_{l}(x)}{2},\quad \Tilde{\sigma}&=\frac{c_{u}(x)-c_{l}(x)}{2}.
\end{align*}
Recall that a generalized CQR score function (\ref{eq:cqr.gen}) characterized by ($\mu_1$, $\mu_2$, $\sigma$) is defined as:
\begin{align}
E_{\mu_1,\mu_2,{\sigma}}(x,y)=\max\left\{{\mu_1}(x)-y,y-\mu_2(x)\right\}/{\sigma}(x),
\end{align}
which leads to the asymptotic conformalized prediction intervals of the form 
\begin{align}
 [{\mu_1}(X)&-{\sigma}(X){Q}_{1-\alpha}(E_{\mu_1,\mu_2,{\sigma}}),{\mu_2}(X)+{\sigma}(X){Q}_{1-\alpha}(E_{\mu_1,\mu_2,{\sigma}})].
\end{align}
Plugging in $\mu_1$, $\mu_2$, $\Tilde{\sigma}$ defined above, we immediately have 
\begin{align*}
    &\mathbbm{P}(Y\in [c_{l}(X),c_{u}(X)])= 1-\alpha\\
    \Longleftrightarrow\hspace{0.5em} & \mathbbm{P}(c_l(X)-Y\leq 0, Y-c_u(X)\leq 0)= 1-\alpha\\
    \Longleftrightarrow\hspace{0.5em}  & \mathbbm{P}\left(\frac{c_l(X)-Y+\Tilde{\sigma}(X)}{\Tilde{\sigma}(X)}\leq 1, \frac{Y-c_u(X)+\Tilde{\sigma}(X)}{\Tilde{\sigma}(X)}\leq 1\right)= 1-\alpha\\
    \Longleftrightarrow\hspace{0.5em}  &Q_{1-\alpha}(E_{\mu_1,\mu_2,{\sigma}}(x,y))=1.
\end{align*}
\end{proof}
\begin{proof}[Proof of {Proposition} \ref{thm:uni.condcov}]
    It suffices to take $c_l(x)=Q_{\alpha/2}(Y|X=x)$, $c_u(x)=Q_{1-\alpha/2}(Y|X=x)$ and apply Lemma \ref{lem:uni}.
\end{proof}

\subsection{Boosting for length: theoretical guarantees}\label{sec:thm.power}
Similar to {Proposition} \ref{thm:uni.condcov}, we show in {Proposition} \ref{thm:uni.eff} below that the generalized Local and CQR score families exhibit the necessary expressiveness to contain the oracle score, achieving optimal length while ensuring valid marginal coverage. 
\begin{prop}[Asymptotic expressiveness]\label{thm:uni.eff}
    Under the assumptions of {Proposition} \ref{thm:uni.condcov}, for any target coverage rate $1-\alpha$, as $n\rightarrow \infty$, the following statements hold true:
    \begin{enumerate}
        \item There exists $(\mu^*,\sigma^*)$ such that the corresponding generalized Local score function (\ref{eq:local.gen}) recovers the {shortest} oracle prediction interval (\ref{eq: oracle.len}).
        \item There exists $(\mu_1^*,\mu_2^*,\sigma^*)$ such that the corresponding generalized CQR score function (\ref{eq:cqr.gen}) recovers the {shortest} oracle prediction interval (\ref{eq: oracle.len}).
    \end{enumerate}
\end{prop}
\begin{proof}[Proof of {Proposition} \ref{thm:uni.eff}]
    It suffices to take $c_l(x)=Q_{l(\alpha)}(Y|X=x)$, $c_u(x)=Q_{u(\alpha)}(Y|X=x)$ and apply Lemma \ref{lem:uni}, where $Q_{l(\alpha)}$ and $Q_{u(\alpha)}$ are the lower and upper limits of the High Density Region defined in (\ref{eq: oracle.len}).
\end{proof}

\subsection{CQR type conformity score boosting} \label{app.sec:cqr.to.local}
A generalized CQR score function (\ref{eq:cqr.gen}) is uniquely defined by a triple ($\mu_1(\cdot),\mu_2(\cdot),\sigma(\cdot)$). We will show how searching for a generalized CQR score can be reduced to searching for a Local generalized score. To begin with, we shall say that score functions are equivalent if they recover identical conformalized prediction intervals.
\begin{definition}\label{def.equiv}
Let $\{X_i,Y_i\}_{i=1}^n$, $(X_{n+1},Y_{n+1})$ be i.i.d. with continuous joint probability density distribution, and let $[n]$ be partitioned into a training set $I_1$ and a calibration set $I_2$. Consider two conformity score functions,  $E_1$ and $E_2$, which produce conformalized prediction intervals $C_1(\cdot)$ and $C_2(\cdot)$, respectively. For any target coverage rate $1-\alpha$, $E_1$ and $E_2$ are equivalent if $C_1(\cdot)=C_2(\cdot)$ when marginal coverage rates $\mathbbm{P}(Y_{n+1}\in C_1(X_{n+1}))$ and $\mathbbm{P}(Y_{n+1}\in C_2(X_{n+1}))$ match.
\end{definition}
Building on this definition, we are now equipped to establish the following equivalences:
\begin{lemma}\label{lemma: equiv}
Under the assumptions of  Definition \ref{def.equiv}, the following statements hold:
\begin{enumerate}
    \item For the CQR-r score function defined in Section \ref{subsec:prelim}, there is an equivalent generalized Local score function characterized by a pair $(\mu(\cdot),\sigma(\cdot))$, where $\mu=(\mu_1+\mu_2)/2$, $\sigma=(\mu_2-\mu_1)/2$.
    \item For any generalized Local score function characterized by the pair $(\mu(\cdot),\sigma(\cdot))$, there is an equivalent generalized CQR score function characterized by a triple $(\mu(\cdot),\mu(\cdot),\sigma(\cdot))$.
\end{enumerate}
\end{lemma}

The proof of the above Lemma is deferred to Section \ref{app.sec:proof.equiv}. Leveraging these equivalences, we carry out the boosted conformal procedure as follows:  first, we initialize a triple $\mu_1^{(0)}=\hat{q}_{\alpha/2}$, $\mu_2^{(0)}=\hat{q}_{1-\alpha/2}$, $\sigma_1^{(0)}=\hat{q}_{1-\alpha/2}-\hat{q}_{\alpha/2}$, which characterizes the CQR-r score function. Next, we find an equivalent generalized Local score function characterized by a pair $(\mu^{(0)}, \sigma^{(0)})$ chosen according to Lemma \ref{lemma: equiv}.   
After $\tau$ boosting rounds, we obtain the boosted pair $(\mu^{(\tau)},\sigma^{(\tau)})$ and the corresponding score function. %
{Finally, we recover an equivalent generalized CQR score function
\begin{align*}
E^{(\tau)}(x,y) = \max\left\{{\mu_1^{(\tau)}}(x)-y,y-{\mu_2}^{(\tau)}(x)\right\}/{\sigma}_1^{(\tau)}(x),
\end{align*}
characterized by the triple $({\mu_1}^{(\tau)},{\mu_2}^{(\tau)},{\sigma_1}^{(\tau)})$ chosen according to Lemma \ref{lemma: equiv}. }

\subsection{Proof of Lemma \ref{lemma: equiv}}\label{app.sec:proof.equiv}
   Recall that the generalized Local score (\ref{eq:local.gen}) characterized by $(\mu,\sigma)$ takes the form
\begin{align}
    E_{\mu,\sigma}(x,y)=\frac{|y-\mu{(x)}|}{\sigma(x)}.
\end{align}
The conformalized prediction interval is given by
\begin{align}
    \left[{\mu}(X)-{Q}_{1-\alpha}(E_{\mu,\sigma},I_2){\sigma}(X), {\mu}(X)+{Q}_{1-\alpha}(E_{\mu,\sigma},I_2){\sigma}(X)\right].
\end{align}
A generalized CQR score function (\ref{eq:cqr.gen}) characterized by ($\mu_1$, $\mu_2$, $\sigma$) is defined as:
\begin{align*}
E_{\mu_1,\mu_2,{\sigma}}(x,y)=\max\left\{{\mu_1}(x)-y,y-\mu_2(x)\right\}/{\sigma}(x),
\end{align*}
which leads to conformalized prediction intervals of the form 
\begin{align*}
 [{\mu_1}(X)&-{\sigma}(X){Q}_{1-\alpha}(E_{\mu_1,\mu_2,{\sigma}},I_2),{\mu_2}(X)+{\sigma}(X){Q}_{1-\alpha}(E_{\mu_1,\mu_2,{\sigma}},I_2)].
\end{align*}
\begin{enumerate}
    \item Plugging in the triple $\mu_1(x)=\hat{q}_{\alpha/2}$, $\mu_2(x)=\hat{q}_{1-\alpha/2}$, $\sigma_1(x)=\hat{q}_{1-\alpha/2}-\hat{q}_{\alpha/2}$, which characterize the CQR-r score function, we have the conformalized prediction interval 
    \begin{align*}
 [\mu_1(X)-\sigma_1(X){Q}_{1-\alpha}(E_{\mu_1,\mu_2,\sigma_1},I_2),\mu_2(X)+\sigma_1(X){Q}_{1-\alpha}(E_{\mu_1,\mu_2,\sigma_1},I_2)].
\end{align*}
Set
    \begin{align*}
    \mu(X)= \frac{\mu_1(X)+\mu_2(X)}{2},\sigma(X)= \frac{\mu_2(X)-\mu_1(X)}{2},
\end{align*}
then the generalized Local conformity score $E_{\mu,\sigma}(x,y)={|y-\mu(x)|}/{\sigma(x)}$ recovers conformalized prediction intervals of the form
\begin{align*}
    &[\mu(X)-\sigma(X)Q(E_{\mu,\sigma},I_2),\mu(X)+\sigma(X)Q(E_{\mu,\sigma},I_2)]\\
=&\left[\frac{\mu_1(X)+\mu_2(X)}{2}-\frac{\mu_2(X)-\mu_1(X)}{2}Q(E_{\mu,\sigma},I_2), \right. \\
&\left.\quad\quad\quad\quad\quad\quad\quad\quad\quad\quad\quad\quad \frac{\mu_1(X)+\mu_2(X)}{2}+\frac{\mu_2(X)-\mu_1(X)}{2}Q(E_{\mu,\sigma},I_2)\right]\\
=&\left[\mu_2(X)-(\mu_2(X)-\mu_1(X))\frac{Q(E_{\mu,\sigma},I_2)-1}{2}, \right. \\
&\left.\quad\quad\quad\quad\quad\quad\quad\quad\quad\quad\quad\quad\mu_1(X)+(\mu_2(X)-\mu_1(X))\frac{Q(E_{\mu,\sigma},I_2)-1}{2}\right]\\
=&\left[\mu_1(X)-\sigma_1(X)\frac{Q(E_{\mu,\sigma},I_2)-1}{2},\mu_2(X)+\sigma_1(X)\frac{Q(E_{\mu,\sigma},I_2)-1}{2}\right].
\end{align*}
From the monotonicity of the interval lengths with respect to the empirical quantiles, we have that the two score functions are equivalent by Definition \ref{def.equiv}.
\item  Let a generalized Local score function be $E_{\mu,\sigma}(x,y)=|y-\mu(x)|/\sigma(x)$. Then it suffices to observe that
\begin{align*}
    |y-\mu(x)|=\max\{y-\mu(z),\mu(x)-y\}.
\end{align*}

\end{enumerate}

\subsection{Additional information on real datasets}\label{sec.app: data}
In Table \ref{tab:data}, we provide the predicted label, dimensions, and source for each dataset. Data cleaning and preprocessing are in accordance with the methods described by Romano et al. \cite{romano2019conformalized}.
\begin{table}[H]
    \centering
    
    \caption{Datasets for our empirical analyses, with the predicted label, number of samples ($n$), and features ($d$).}
    \begin{tabular}{c c c c c}
    \toprule
    Name & Label  &$n$ &$d$ &Source\\
     
    \hline
    
    bike & bike rental counts & 10886 & 18&\cite{misc_bike_sharing_dataset_275} \\
    bio &deviation of predicted from native protein structure & 45730 & 9&\cite{misc_physicochemical_properties_of_protein_tertiary_structure_265}\\
    blog & number of comments in the next 24 hours&52397&280&\cite{misc_blogfeedback_304}\\
    community & crime rate per community&1994&100&\cite{misc_communities_and_crime_183}\\
    concrete & concrete compressive strength&1030&8&\cite{misc_concrete_compressive_strength_165}\\
    facebook-1 &Facebook comment volume & 40948& 53&\cite{misc_concrete_compressive_strength_165}\\
    facebook-2 &Facebook comment volume &81311&53&\cite{misc_facebook_comment_volume_dataset_363}\\
    meps-19 &utilization of medical services&15785&139&\cite{misc_meps19}\\
    meps-20 &utilization of medical services&17541&139&\cite{misc_meps20}\\
    meps-21 &utilization of medical services&15656&139&\cite{misc_meps21}\\
    star &total student test scores up to the third grade&2161&39&\cite{misc_star}\\
    \bottomrule
    \end{tabular}
    \label{tab:data}
\end{table}
All datasets, except for the meps and star data sets, are licensed under CC-BY 4.0. The Medical Expenditure Panel Survey (meps) data is subject to copyright and usage rules. The licensing status of the star dataset could not be determined.

\subsection{Experimental Setup}\label{sec:model.hyperparams}
 In each dataset, we randomly hold out $20\%$ as test data.
The remaining data is divided into a training set and a calibration set, each taking up a proportion of $\gamma$ and $1-\gamma$. We explore training ratios $\gamma$ ranging from $10\%$ to $90\%$. Results corresponding to the optimal value of the hyperparameter $\gamma$ are recorded in Table \ref{tab:condcov}, following the practice of Sesia er al. \cite{Sesia_2020}.

In the training stage, we employ the random forest regressor from Python's scikit-learn package to learn the baseline Local score function. The hyperparameters are the package defaults, except for the total number of trees, which we set to 1000, and the minimum number of samples required at a leaf node, which we set to 40, as recommended by Romano et al.~\cite{romano2019conformalized}. For the baseline CQR score function, we adopt a black-box neural network quantile regressor with three fully connected layers and ReLU non-linearities, following the practice of Sesia et al.~\cite{Sesia_2020}. In the boosting stage, we set the hyper-parameters $\tau_1$, $\tau_2$ in the approximated loss (\ref{eq:proxy}) to 50. The approximated loss is then passed to the Gradient Boosting Machine from Python's XGBoost package along with a base conformity score. We set the maximum tree depth to 1 to avoid overfitting and perform cross-validation for the number of boosting rounds, as outlined in Section \ref{subsec:boosted.cp}.  All other hyperparameters are set to package defaults. 

{All experiments were conducted on a dual-socket AMD EPYC 7502 32-Core Processor system, utilizing 8 of its 128 CPUs each time.} The runtime for each dataset and random seed varies by dataset size, ranging from 10 minutes to 5 hours. 

\subsection{Additional results and error bars}\label{sec.app:error.bars}

In Tables \ref{tab:condcov.local.full} to \ref{tab:len.cqr.full}, we present additional results on marginal coverage, maximum conditional coverage deviation {($\ell_M$)}, and average interval length {($\ell_P$)} for each real dataset (including those not reported in Tables \ref{tab:condcov} and \ref{tab:table1}), both before and after boosting. Notably, in each case, boosting is applied to optimize either conditional coverage or average interval length. As a result, the non-targeted characteristic may or may not improve after boosting.

 \begin{table}[H]
  \begin{center}
  
    \caption{Additional information on conformalized intervals obtained before and after boosting for conditional coverage with the Local conformity score as baseline. The target miscoverage rate is set to $\alpha=10\%$.}
    \vspace{0.1cm}
    \label{tab:condcov.local.full}
    {\small
    \begin{tabular}{S S S S S S S}
   
      \hline
  \multicolumn{1}{c}{\multirow{2}{*}{Dataset}}& \multicolumn{2}{c}{$\ell_L$} & \multicolumn{2}{c}{$\ell_M(\%)$} & \multicolumn{2}{c}{Marginal Cov.(\%)} \\ 
  \multicolumn{1}{c}{} & \multicolumn{1}{c}{Local} & $\text{Boosted}$& \multicolumn{1}{c}{Local} & $\text{Boosted}$& \multicolumn{1}{c}{Local} & $\text{Boosted}$ \\ 
  
  \hline      
  $\text{bike}$& ${1.775}$& ${2.201}$& {10.979}&{5.638}& {89.927}& {89.646}\\
  
  bio& {1.602}&{1.614}& {5.303} & {4.862}&{90.024}&{90.093}\\
 blog & ${2.080}$&{3.403}& {51.353}& {48.591}&{89.995}&{90.040}\\
  community& {1.824}&{10.759}& {25.755} & {13.466}&{90.376}&{89.549}\\
  $\text{concrete}$& {1.058}&{1.062}& {10.740}&{8.763}&{90.583}&{91.359}\\
 $\text{facebook-1}$& {1.896}&{4.790}& {26.020}& {25.917}&{90.201}&{90.056}\\
 $\text{facebook-2}$ & {1.881}&{2.273}& {42.807} & {42.437}&{89.966}&{89.989}\\
  $\text{meps-19}$ & {2.074}&{2.926}& {15.357}& {5.656}&{90.120}&{90.497}\\
  $\text{meps-20}$ & {2.102}&{2.778}& {16.939}& {6.998}&{89.963}&{90.031}\\
  $\text{meps-21}$& {2.069}&{2.537}& {17.627} &{7.832}&{90.064}&{90.054}\\
  star& {0.189}&{0.179}& {9.658}&{9.348}&{90.808}&{90.831}\\
      
      \hline
    \end{tabular}}
  \end{center}
\end{table}

 \begin{table}[H]
  \begin{center}
  
    \caption{Additional information on conformalized intervals obtained before and after boosting for conditional coverage with the CQR conformity score as baseline. The target miscoverage rate is set to $\alpha=10\%$.}
    \vspace{0.1cm}
    \label{tab:condcov.cqr.full}
    {\small
    \begin{tabular}{S S S S S S S}
 \hline
  \multicolumn{1}{c}{\multirow{2}{*}{Dataset}}& \multicolumn{2}{c}{$\ell_L$} & \multicolumn{2}{c}{$\ell_M$(\%)} & \multicolumn{2}{c}{Marginal Cov.(\%)}\\ 
  \multicolumn{1}{c}{} & \multicolumn{1}{c}{CQR} & $\text{Boosted}$& \multicolumn{1}{c}{CQR} & $\text{Boosted}$& \multicolumn{1}{c}{CQR} & $\text{Boosted}$\\ 
  \hline
  $\text{bike}$&\text{0.555}&{0.540}&{4.934}&{4.925}&{90.073}&{90.184}\\
  bio&\text{1.518}&{1.515}&{5.069}&\text{4.700}&{89.841}&{89.853}\\
 blog &\text{1.761}&{1.766}&{27.760}&{26.836}&{90.222}&{90.244}\\
  community&{1.718}&{1.740}&{12.688}&{12.105}&{90.340}&{90.194}\\
  $\text{concrete}$&\text{0.484}&{0.489}&{9.039}&{8.265}&{90.451}&{90.652}\\
 $\text{facebook-1}$&{1.374}&{1.371}&{13.407}&{13.255}&{90.465}&{90.247}\\
 $\text{facebook-2}$ &{1.465}&{1.409}&{18.257}&{18.002}&{89.763}&{90.001}\\
  $\text{meps-19}$&{2.784}&{2.784}&{5.507}&{5.507}&{90.257}&{90.257}\\
  $\text{meps-20}$ &{2.769}&{2.743}&{7.614}&\text{7.184}&{89.991}&{90.006}\\
  $\text{meps-21}$&{2.834}&{2.815}&{8.16}&\text{8.067}&{90.169}&{90.067}\\
  star&{0.199}&{0.209}&{9.728}&{9.630}&{91.085}&{91.339}\\
    

      \hline
    \end{tabular}}
  \end{center}
\end{table}

\begin{table}[H]
  \begin{center}
  
    \caption{Additional information on conformalized intervals obtained before and after boosting for length with the Local conformity score as baseline. The target miscoverage rate is set to $\alpha=10\%$.}
    \vspace{0.1cm}
    \label{tab:len.local.full}
    {\small
    \begin{tabular}{c S S S S S S}
      \hline
      \multicolumn{1}{c}{\multirow{2}{*}{Dataset}}& \multicolumn{2}{c}{$\ell_L$} & \multicolumn{2}{c}{$\ell_M$(\%)} & \multicolumn{2}{c}{Marginal Cov.(\%)} \\ 

      \multicolumn{1}{c}{} & \multicolumn{1}{c}{Local} & $\text{Boosted}$& \multicolumn{1}{c}{Local} & $\text{Boosted}$& \multicolumn{1}{c}{Local} & $\text{Boosted}$ \\ 
      
      \hline      
      $\text{bike}$& {1.775}&{1.360}& {22.590}& {19.616}& {89.927}& {89.862}\\
bio&{1.562} & {1.514}& {5.995}&{5.791}&{89.937}&{89.962}\\
blog &  {2.056}& {0.972} & {52.440}&{54.858}&{89.988}&{89.978}\\
community&{1.728} & {1.678}& {26.066}&{24.822}&{89.499}&{89.323}\\
concrete & {1.010}&{0.698} & {11.029}&{10.518}&{90.631}&{90.728}\\
$\text{facebook-1}$& {1.896}& {1.384}& {26.020}&{34.259}&{90.201}&{89.944}\\
$\text{facebook-2}$ & {1.854} & {1.363} & {42.624}&{50.697}&{90.020}&{89.972}\\
$\text{meps-19}$ & {2.070}& {1.685}& {18.626}&{14.623}&{90.054}&{90.070}\\
$\text{meps-20}$ & {2.081}& {1.836}& {17.897}&{14.643}&{89.869}&{89.849}\\
$\text{meps-21}$&{2.063} &{1.795} & {18.795}&{13.324}&{89.914}&{89.920}\\
star& {0.179}&{0.179} & {9.976}&{9.407}&{90.901}&{90.577}\\
      \hline
    \end{tabular}}
  \end{center}
\end{table}

 \begin{table}[H]
  \begin{center}
  
    \caption{Additional information on conformalized intervals obtained before and after boosting for length with the CQR conformity score as baseline. The target miscoverage rate is set to $\alpha=10\%$.}
    \vspace{0.1cm}
    \label{tab:len.cqr.full}
    {\small
    \begin{tabular}{c S S S S S S}
      \hline
      \multicolumn{1}{c}{\multirow{2}{*}{Dataset}}& \multicolumn{2}{c}{$\ell_L$} & \multicolumn{2}{c}{$\ell_M$(\%)} & \multicolumn{2}{c}{Marginal Cov.(\%)}\\ 

      \multicolumn{1}{c}{} & \multicolumn{1}{c}{CQR} & $\text{Boosted}$& \multicolumn{1}{c}{CQR} & $\text{Boosted}$& \multicolumn{1}{c}{CQR} & $\text{Boosted}$\\ 
      \hline
     \text{bike} & {0.553} & {0.489} & \text{9.530} & {10.092} & {90.041} & {90.418} \\
\text{bio} & {1.516} & {1.468} & \text{5.408} & {5.265} & {89.670} & {89.880} \\
\text{blog} & {1.445} & {1.434} & \text{29.875} & {34.011} & {90.149} & {90.260} \\
\text{community} & {1.693} & {1.699} & {14.006} & {13.536} & {89.499} & {89.699} \\
\text{concrete} & {0.391} & {0.393} & \text{10.342} & {10.700} & {88.932} & {89.223} \\
\text{facebook-1} & {1.198} & {1.073} & {20.132} & {30.901} & {89.937} & {90.013} \\
\text{facebook-2} & {1.200} & {1.075} & {29.233} & {34.475} & {90.035} & {90.010} \\
\text{meps-19} & {2.554} & {2.136} & {10.357} & {11.285} & {90.228} & {90.399} \\
\text{meps-20} & {2.667} & {2.357} & {10.826} & {10.720} & {89.875} & {89.838} \\
\text{meps-21} & {2.585} & {2.105} & {10.968} & {10.565} & {89.946} & {89.863} \\
\text{star} & {0.195} & {0.194} & {9.982} & {10.142} & {91.455} & {91.432} \\

      \hline
    \end{tabular}}
  \end{center}
\end{table}

We have previously reported the evaluated losses $\ell_M$ and $\ell_L$ for each dataset, averaged over ten random seeds. Tables \ref{tab:error.bar.cond} and \ref{tab:error.bar.power} below detail the distribution of these evaluations, providing the mean, $10\%$ quantile, and $90\%$ quantile for the test set deviations in conditional coverage ($\ell_M$) and {average interval length} ($\ell_L$). These statistics are derived from 110 test set evaluations across 11 datasets and 10 random training-test splits. We opt to report empirical quantiles instead of standard deviations due to the asymmetric and non-Gaussian nature of the data.
\begin{table}[H]
  \begin{center}
    \caption{Distribution of the test set conditional coverage deviation $\ell_M$ evaluated on various conformalized prediction intervals across 11 datasets and 10 random training-test splits.
    }
    \label{tab:error.bar.cond}
    \vspace{0.1cm}

    \begin{tabular}{S S S S S}
    \hline
    \multicolumn{5}{c}{Max. Conditional Coverage Deviation ($\%$), target miscoverage $\alpha = 10\%$}\\
      \hline
      \multicolumn{1}{c}{\multirow{2}{*}{Statistics}}& \multicolumn{2}{c}{Method} & \multicolumn{2}{c}{Method} \\

      \multicolumn{1}{c}{} & \multicolumn{1}{c}{Local}& \multicolumn{1}{c}{Boosted} & $\text{CQR}$& \multicolumn{1}{c}{Boosted}\\ 
      \hline
     $\text{mean}$& 21.139920528346673$\%$  &16.319050912454394$\%$ &11.10606016582242$\%$  &10.770604958169429$\%$ \\
      $\text{10}\%\text{ quantile}$  & 7.267063921993501$\%$&4.604106587687164$\%$ &4.890353796944252$\%$ &4.909938138939401$\%$ \\
    $\text{90}\%\text{ quantile}$& 47.832351089570047$\%$ &44.712387042223756$\%$&22.585248360171437$\%$  &18.69673306772906$\%$ \\
      \hline
    \end{tabular}
  \end{center}
\end{table}

\begin{table}[H]
  \begin{center}
    \caption{Distribution of the test set {average interval length} $\ell_L$ evaluated on various conformalized prediction intervals across 11 datasets and 10 random training-test splits.
    }
    \label{tab:error.bar.power}
    \vspace{0.1cm}

    \begin{tabular}{S S S S S}
    \hline
    \multicolumn{5}{c}{Average Length, target miscoverage $\alpha = 10\%$}\\
      \hline
      \multicolumn{1}{c}{\multirow{2}{*}{Statistics}}& \multicolumn{2}{c}{Method} & \multicolumn{2}{c}{Method} \\

      \multicolumn{1}{c}{} & \multicolumn{1}{c}{Local}& \multicolumn{1}{c}{Boosted} & $\text{CQR}$& \multicolumn{1}{c}{Boosted}\\ 
      \hline
     $\text{mean}$& 1.6774321511111188  &1.3167571673500214  &  1.4826891378899083 & 1.3194553659090906 \\
      $\text{10}\%\text{ quantile}$  & 0.9502768404499357 & 0.5126313678434842 &  0.345673884 & 0.350588561 \\
    $\text{90}\%\text{ quantile}$& 2.0824626928427494  & 1.8288196087095183 &  2.6548226400000003 & 2.21005927 \\
      \hline
    \end{tabular}
  \end{center}
\end{table}

\subsection{Experiments under different miscoverage rates}\label{subsec:diff.miscov}
In Tables \ref{tab:len.local.full.5} and \ref{tab:len.local.full.20}, we illustrate the performance of boosting for length with the Local conformity score as baseline with miscoverage rates set to $5\%$ and $20\%$, respectively.
\begin{table}[H]
  \begin{center}
  
    \caption{Additional information on conformalized intervals obtained before and after boosting for length with the Local conformity score as baseline. The target miscoverage rate is set to $\alpha=5\%$.}
    \vspace{0.1cm}
    \label{tab:len.local.full.5}
    {\small
    \begin{tabular}{S S S S S S S S S}
      \hline
      \multicolumn{1}{c}{\multirow{2}{*}{Dataset}}& \multicolumn{3}{c}{$\ell_L$} & \multicolumn{2}{c}{$\ell_M$(\%)} & \multicolumn{2}{c}{Marginal Cov.(\%)} \\ 

      \multicolumn{1}{c}{} & \multicolumn{1}{c}{Local} & $\text{Boosted}$&\text{Improvement}& \multicolumn{1}{c}{Local} & $\text{Boosted}$& \multicolumn{1}{c}{Local} & $\text{Boosted}$ \\ 
      
      \hline      
$\text{bike}$& {2.523}&{1.828}& {-27.54\%}& {11.553}& {11.626}& {94.927}& {94.972}\\
bio&{1.881} & {1.801}& {-4.25\%}& {5.013}&{4.694}&{94.857}&{94.802}\\
blog &  {4.217}& {1.923} & {-54.39\%}& {37.976}&{34.505}&{95.027}&{95.048}\\
community&{2.578} & {2.217}& {-13.99\%}& {16.406}&{14.554}&{95.414}&{94.812}\\
$\text{concrete}$& {1.129}&{0.838} & {-25.80\%}& {8.55}&{8.086}&{94.709}&{94.515}\\
$\text{facebook-1}$& {2.667}& {2.262}& {-15.18\%}& {20.85}&{20.052}&{95.104}&{95.049}\\
$\text{facebook-2}$ & {2.441} & {2.092} & {-14.30\%}& {41.645}&{39.943}&{94.941}&{95.016}\\
$\text{meps-19}$ & {3.618}& {2.825}& {-21.92\%}& {10.515}&{9.178}&{95.005}&{95.008}\\
$\text{meps-20}$ & {3.951}& {3.049}& {-22.84\%}& {11.043}&{9.693}&{94.907}&{94.922}\\
$\text{meps-21}$&{4.030} &{3.033} & {-24.74\%}& {8.918}&{8.910}&{95.061}&{95.026}\\
star& {0.207}&{0.207} & {-0.11\%}& {6.57}&{6.906}&{95.358}&{95.289}\\
      
      \hline
    \end{tabular}}
  \end{center}
\end{table}

\begin{table}[H]
  \begin{center}
  
    \caption{Additional information on conformalized intervals obtained before and after boosting for length with the Local conformity score as baseline. The target miscoverage rate is set to $\alpha=20\%$.}
    \vspace{0.1cm}
    \label{tab:len.local.full.20}
    {\small
    \begin{tabular}{S S S c S S S S S}
      \hline
      \multicolumn{1}{c}{\multirow{2}{*}{Dataset}}& \multicolumn{3}{c}{$\ell_L$} & \multicolumn{2}{c}{$\ell_M$(\%)} & \multicolumn{2}{c}{Marginal Cov.(\%)} \\ 

      \multicolumn{1}{c}{} & \multicolumn{1}{c}{Local} & $\text{Boosted}$&\text{Improvement}& \multicolumn{1}{c}{Local} & $\text{Boosted}$& \multicolumn{1}{c}{Local} & $\text{Boosted}$ \\ 
      
      \hline   
      \text{bike} & 1.312 & 0.983 & -25.07\% & 24.649 & 28.561 & 79.403 & 80.152 \\
\text{bio} & 1.248 & 1.213 & -2.80\% & 10.838 & 10.540 & 79.729 & 79.594 \\
\text{blog} & 1.904 & 0.513 & -73.07\% & 45.813 & 64.360 & 79.756 & 79.920 \\
\text{community} & 1.337 & 1.234 & -7.71\% & 25.047 & 28.831 & 79.674 & 79.674 \\
\text{concrete} & 0.833 & 0.549 & -34.10\% & 13.831 & 20.177 & 80.728 & 80.825 \\
\text{facebook-1} & 1.624 & 0.747 & -54.02\% & 29.384 & 47.994 & 80.190 & 79.731 \\
\text{facebook-2} & 1.580 & 0.749 & -52.60\% & 36.488 & 62.070 & 79.890 & 79.930 \\
\text{meps-19} & 1.821 & 1.020 & -44.03\% & 20.561 & 24.763 & 80.326 & 80.013 \\
\text{meps-20} & 1.843 & 1.095 & -40.60\% & 19.137 & 24.194 & 79.684 & 79.823 \\
\text{meps-21} & 1.831 & 1.064 & -41.86\% & 18.227 & 23.172 & 80.674 & 80.057 \\
\text{star} & 0.142 & 0.141 & -0.71\% & 15.420 & 14.517 & 80.647 & 80.370 \\
      \hline
    \end{tabular}}
  \end{center}
\end{table}

\subsection{Training a gradient boosting algorithm with our custom loss functions}\label{subsec:compare.optimal.learning}
Our proposed boosted conformal procedure serves as a post-training step designed to refine the conformity score $E(\cdot,\cdot;f)$ obtained during model training. This procedure can leverage pre-trained models when available. In the absence of pre-trained models, we can alternatively train a gradient boosting algorithm directly using our custom loss functions. For example, in the context of the local score, we may initialize with $\mu^{(0)}=0$, $\sigma^{(0)}=1$ and then apply our boosted conformal procedure.  As discussed in Section \ref{sec:discussion}, this approach is flexible enough to replace gradient boosting with any gradient-based algorithm, such as neural networks, trained under our custom loss functions. This framework aligns with that of \cite{stutz2022learning}, which introduces a neural network trained with a custom loss function to minimize the average prediction set size in classification tasks.

In Figure \ref{fig:compare.optimal.class}, we compare the performance of the two approaches optimizing for average interval length, searching within the generalized Local score family $\mathcal{F}$ defined in (\ref{eq:local.gen}). The primary distinction between the two procedures lies in the initialization: the first approach employs $\mu^{(0)}=0$, $\sigma^{(0)}=1$, while the second derives $\mu^{(0)}$ and $\sigma^{(0)}$ from a trained random forests model. We run the experiments on the meps-19 dataset and compare the performance across different splits of the training and calibration data. In this context, the percentage of training data refers to the proportion of training data within the combined training and calibration datasets. Our results indicate that cross-validation selects a greater number of boosting iterations when we directly train the gradient boosting algorithm, resulting in longer runtime. However, the average interval length and maximum conditional coverage deviation after boosting are notably smaller for the boosted conformal procedure we introduced in this paper.

\begin{figure}
    \centering
    \includegraphics[width=\linewidth]{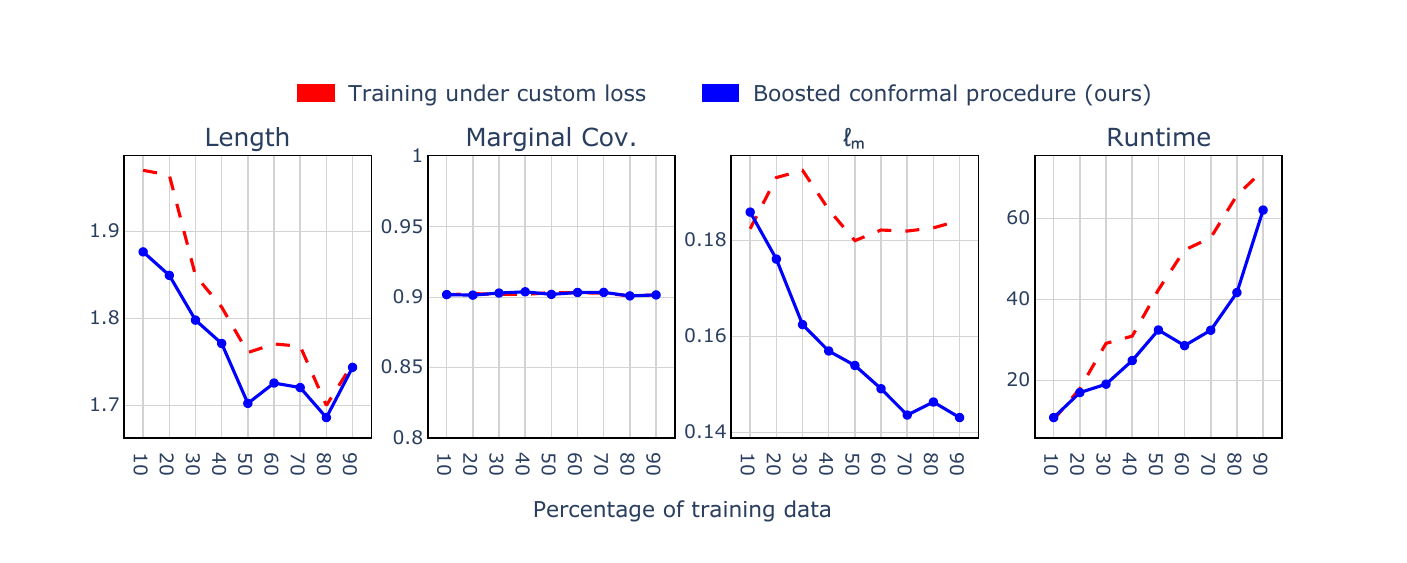}
    \vspace{-0.6cm}
    \caption{Comparison of boosted interval length, marginal coverage, maximum conditional coverage deviation ($\ell_M$), and runtime between direct training of a gradient-based algorithm (red) and boosting on a pre-trained conformity score (blue).}
    \label{fig:compare.optimal.class}
\end{figure}

\subsection{Selecting optimal boosting rounds via hold-out validation set}\label{subsec:validation}
In our boosted conformal procedure, we use cross-validation on the training set to determine the optimal number of boosting rounds, a process that can be time-consuming. An alternative approach is to hold out a fraction of the training set for validation. While more computationally efficient, this method introduces a trade-off: a smaller validation set can lead to greater variability in prediction intervals and model performance, whereas a larger validation set may reduce the effective training set size, potentially limiting the model’s performance. To explore this trade-off, we conduct experiments on the bike dataset, optimizing for prediction interval length. We compare performance across two settings: 5-fold cross-validation and a hold-out validation set, with the training-to-validation set ratio ranging from 1:1 to 8:1. For each setting, we run 100 experiments, recording the average boosted length, the standard deviation of boosted lengths, and the average runtime. The results are shown in Figure \ref{fig:validation}.
\begin{figure}
    \centering

    \includegraphics[width=\linewidth,trim={0cm 1cm 0cm 1cm},clip]{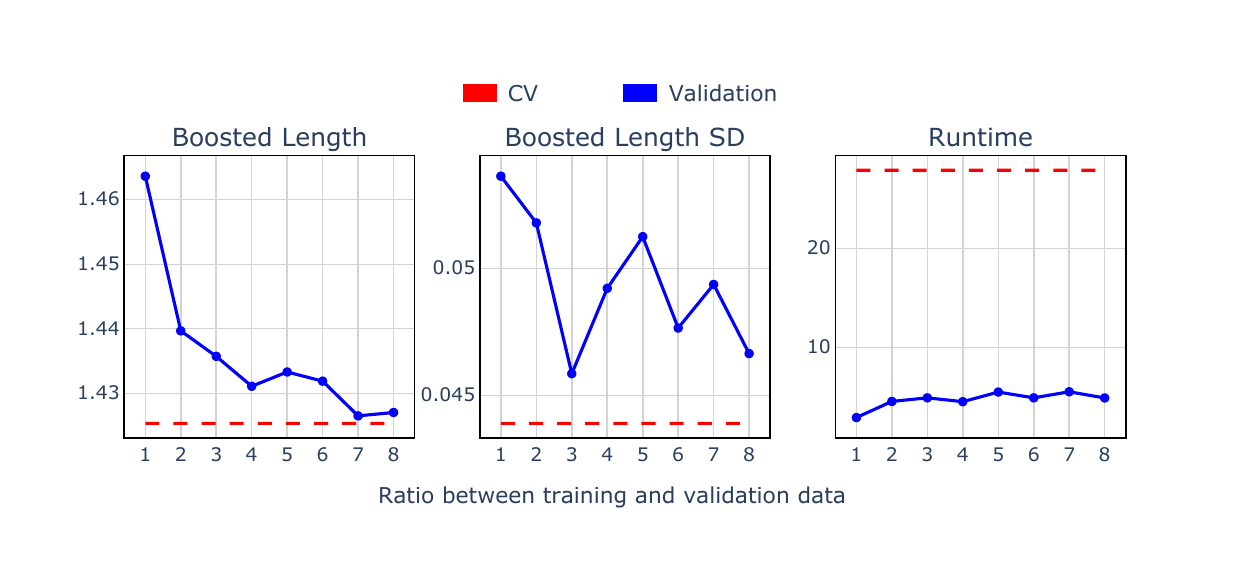}
    \caption{Comparison of the average boosted length, standard deviation of boosted length, and average runtime of the boosting procedure when selecting the optimal number of boosting rounds using 5-fold cross-validation versus a hold-out validation set of varying sizes. Experiments are conducted on the bike dataset, with a target miscoverage rate of $\alpha = 10\%$.}
    \label{fig:validation}
\end{figure}
\subsection{Additional figures on individual datasets}
In this section, we present a series of supplementary figures. First, we showcase the improvements in conditional coverage achieved through the boosted procedure for each benchmark dataset. Figure \ref{fig:supp.cond.meps} details results for datasets meps-20 and meps-21. Figure \ref{fig:supp.cond.com} details results for datasets community, bike, and concrete. 

Next, we illustrate enhanced interval lengths. Figure \ref{fig:supp.len.meps} details results for datasets meps-20, meps-21, and bike. Figure \ref{fig:supp.len.facebook} details results for datasets facebook-1, facebook-2, and concrete. Finally, we demonstrate in Figure \ref{fig:cv} how cross-validating the number of boosting rounds effectively prevents the gradient boosting algorithm from overfitting. 

\begin{figure}[H]
\centering 
 \begin{subfigure}[b]{\textwidth}
         \centering
\includegraphics[trim={0cm 1cm 0cm 0cm},clip,width=\textwidth]{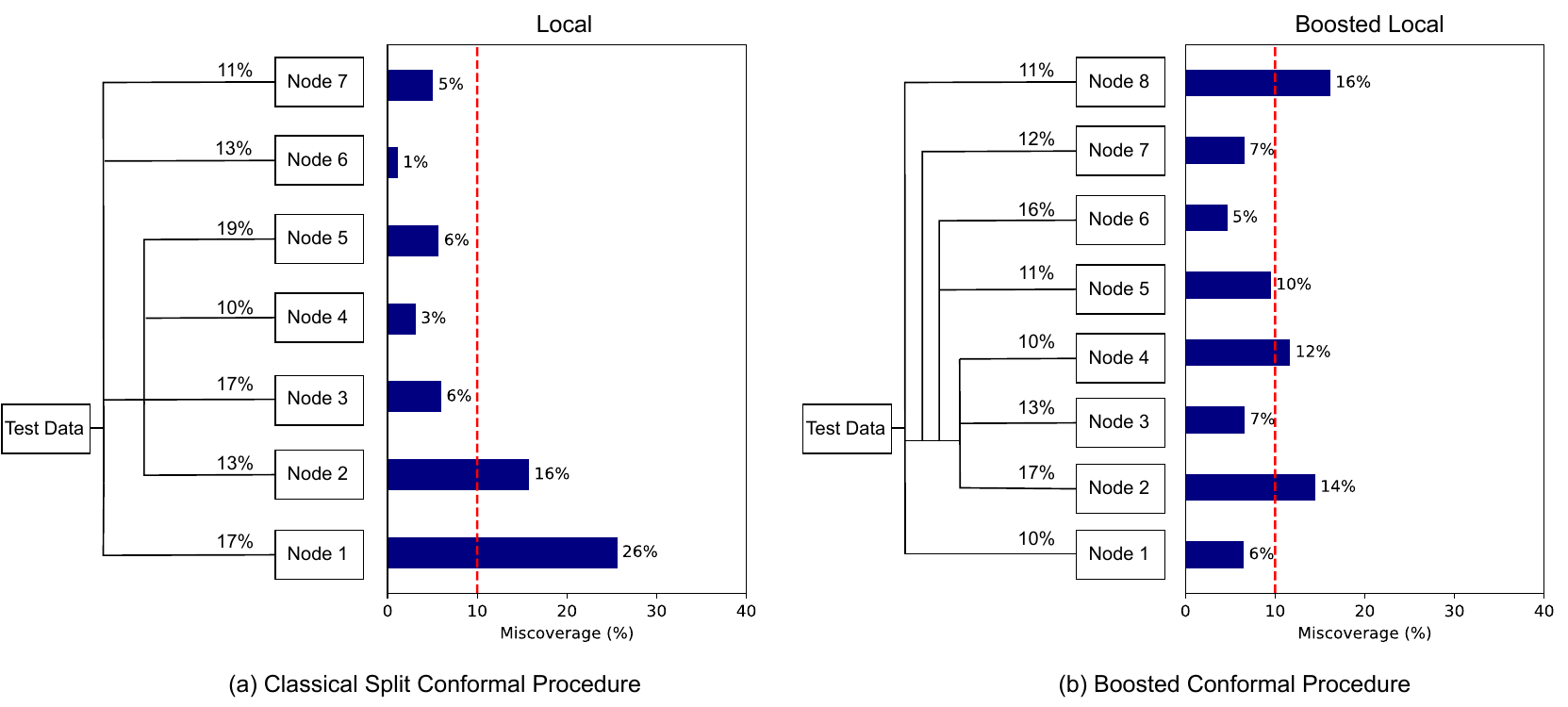}
         \caption{{meps-20} dataset.}
     \end{subfigure}
     \hfill
     \begin{subfigure}[b]{\textwidth}
         \centering
         \includegraphics[trim={0cm 1cm 0cm 0cm},clip,width=\textwidth]{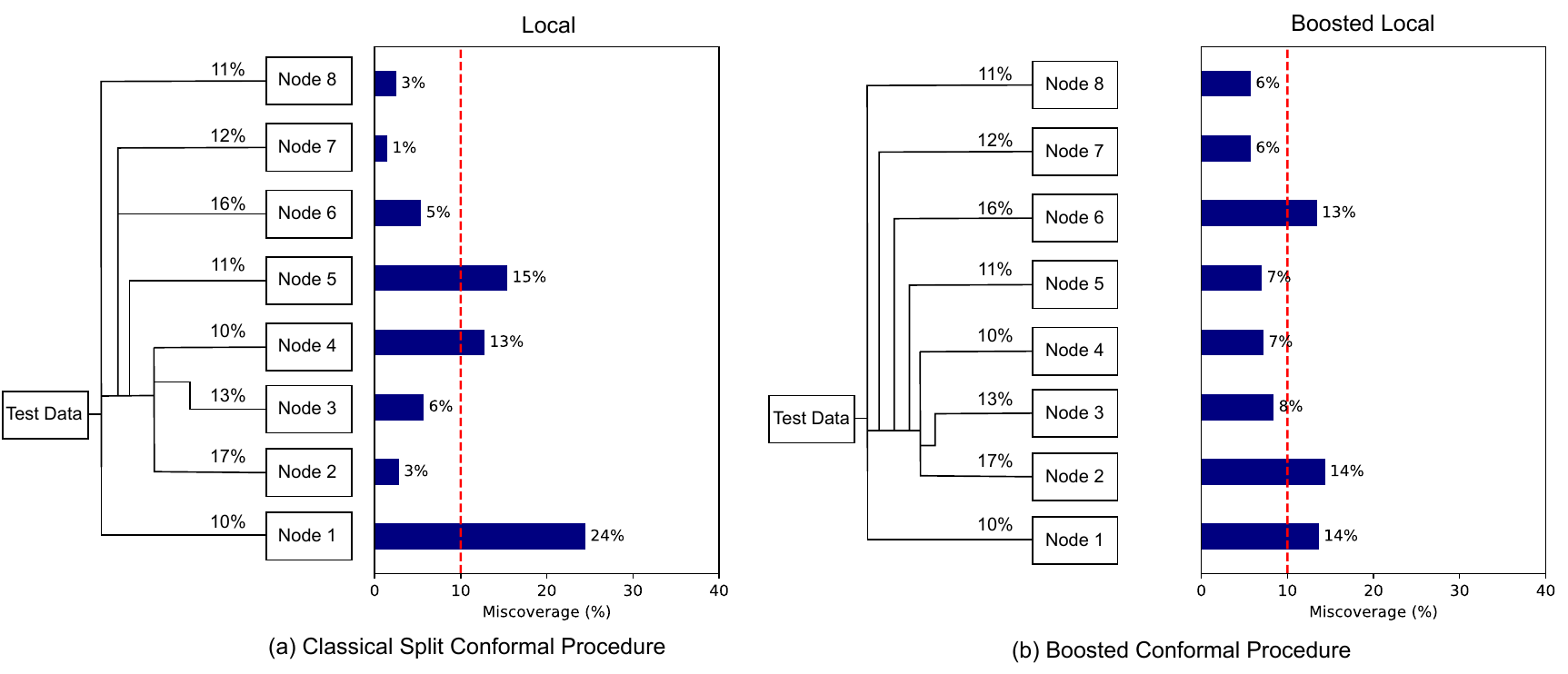}
         \caption{{meps-21} dataset.}
     \end{subfigure}
    \caption{See the caption of Figure \ref{fig:condcov.local.meps19} for details.}
    \label{fig:supp.cond.meps}
\end{figure}

\begin{figure*}[p] 
\centering
\vspace*{\fill}
 \begin{subfigure}[b]{\textwidth}
         \centering
\includegraphics[trim={0cm 1cm 0cm 0cm},clip,width=\textwidth]{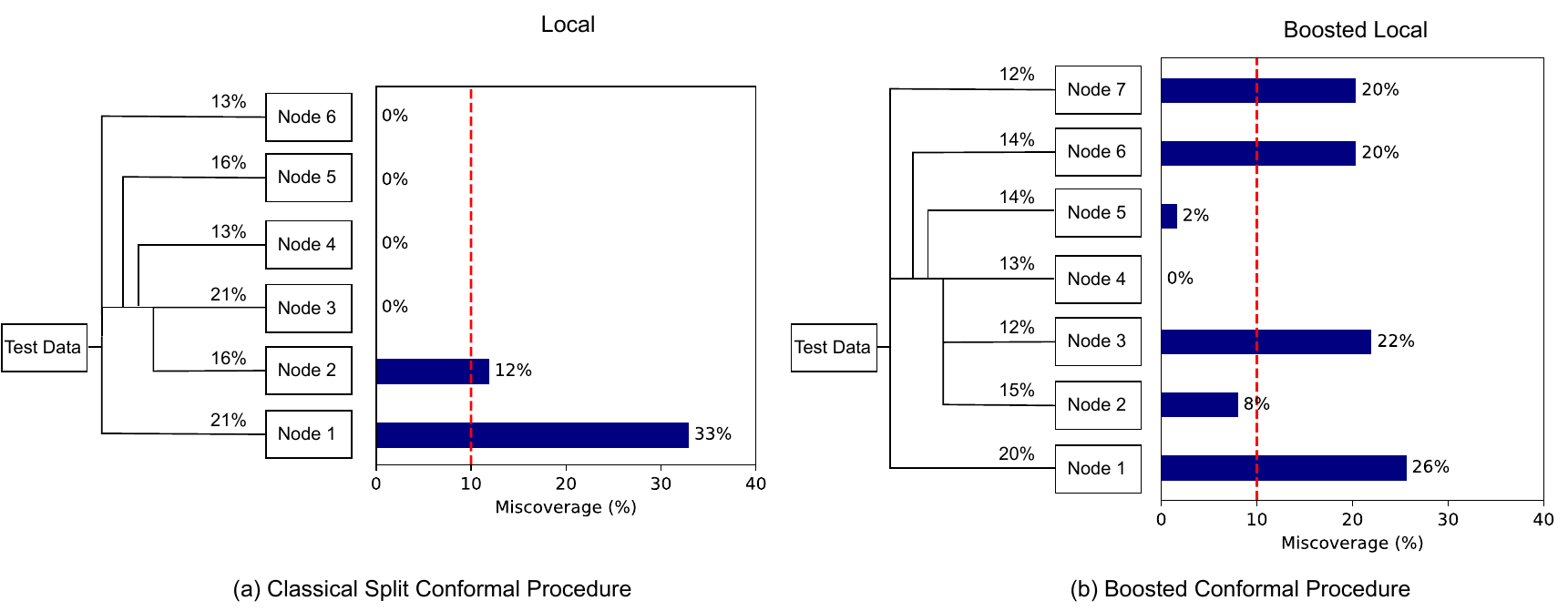}
         \caption{{community} dataset.}
     \end{subfigure}
     \hfill
     \begin{subfigure}[b]{\textwidth}
         \centering
         \includegraphics[trim={0cm 1cm 0cm 0cm},clip,width=\textwidth]{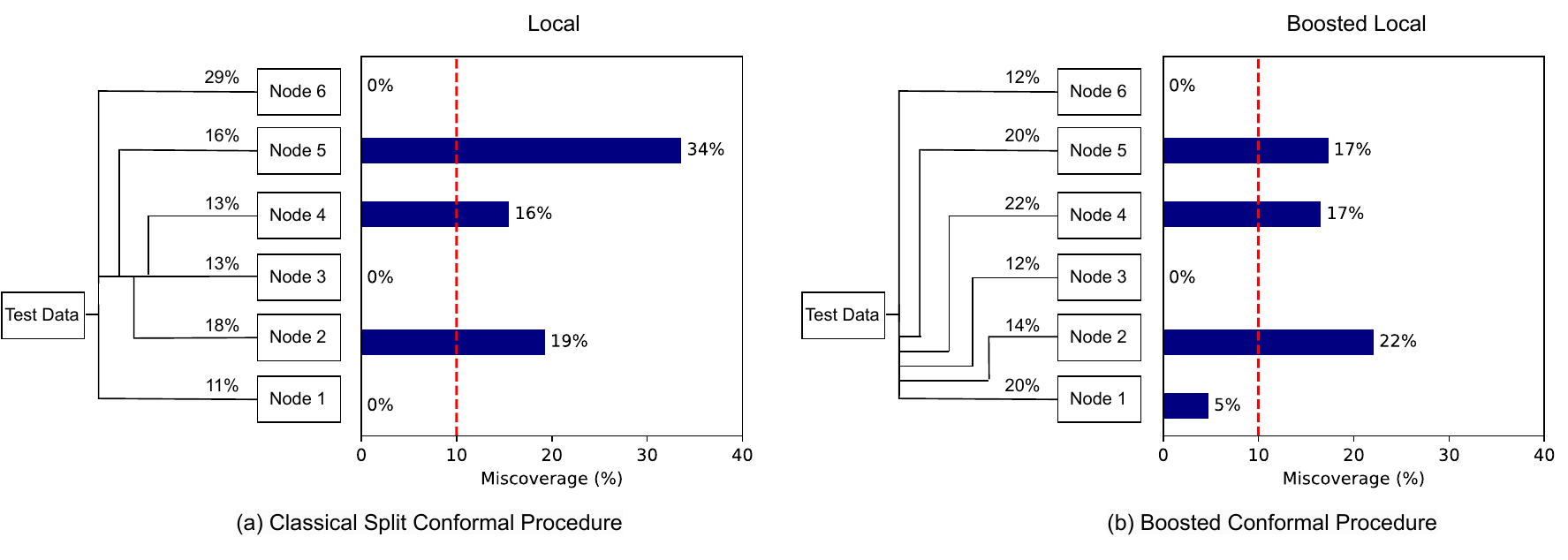}
         \caption{{bike} dataset.}
     \end{subfigure}
      \hfill
     \begin{subfigure}[b]{\textwidth}
         \centering
         \includegraphics[trim={0cm 1cm 0cm 0cm},clip,width=\textwidth]{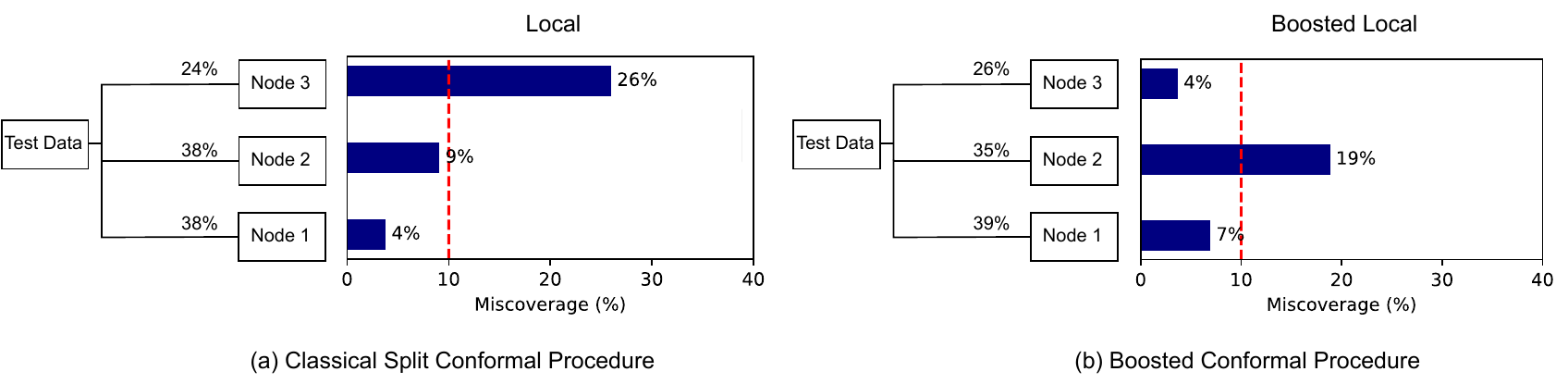}
         \caption{{concrete} dataset.}
     \end{subfigure}

    \caption{See the caption of Figure \ref{fig:condcov.local.meps19} for details.}
    \label{fig:supp.cond.com}
\vspace*{\fill}
\end{figure*}

\begin{figure*}[p] 
\centering
\vspace*{\fill}
 \begin{subfigure}[b]{\textwidth}
         \centering
\includegraphics[width=\textwidth]{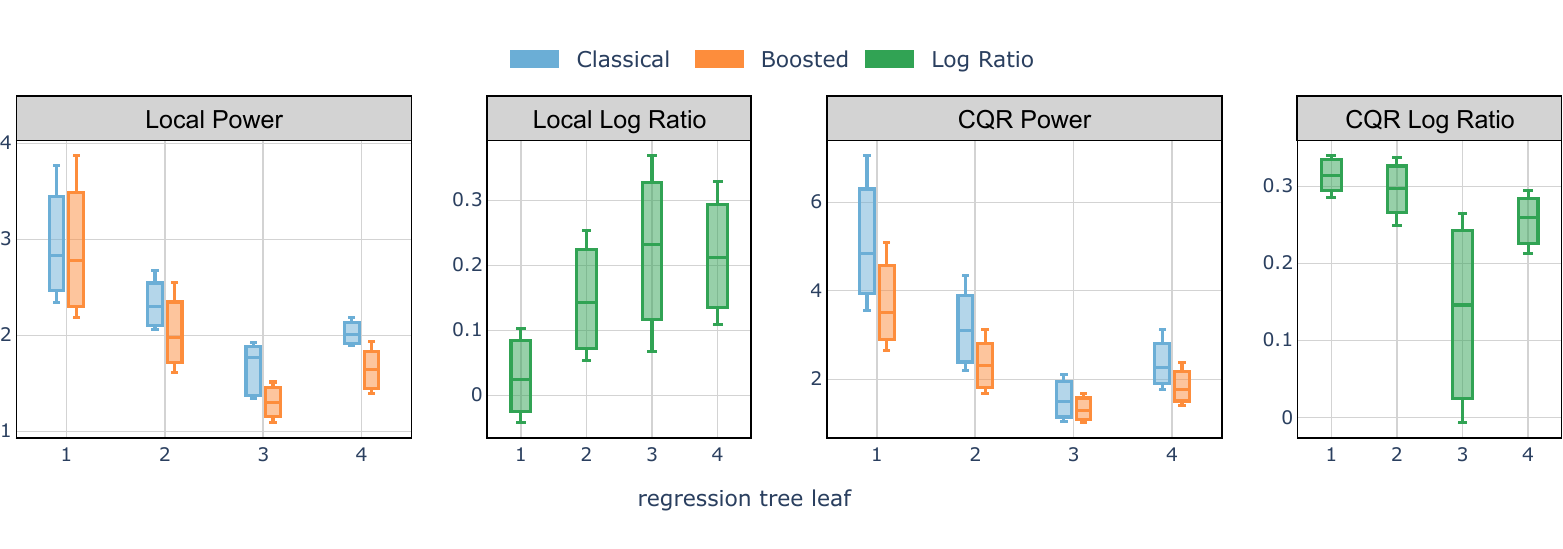}
         \caption{{meps-20} dateset. 
         }
     \end{subfigure}
     \hfill
     \begin{subfigure}[b]{\textwidth}
         \centering
         \includegraphics[width=\textwidth]{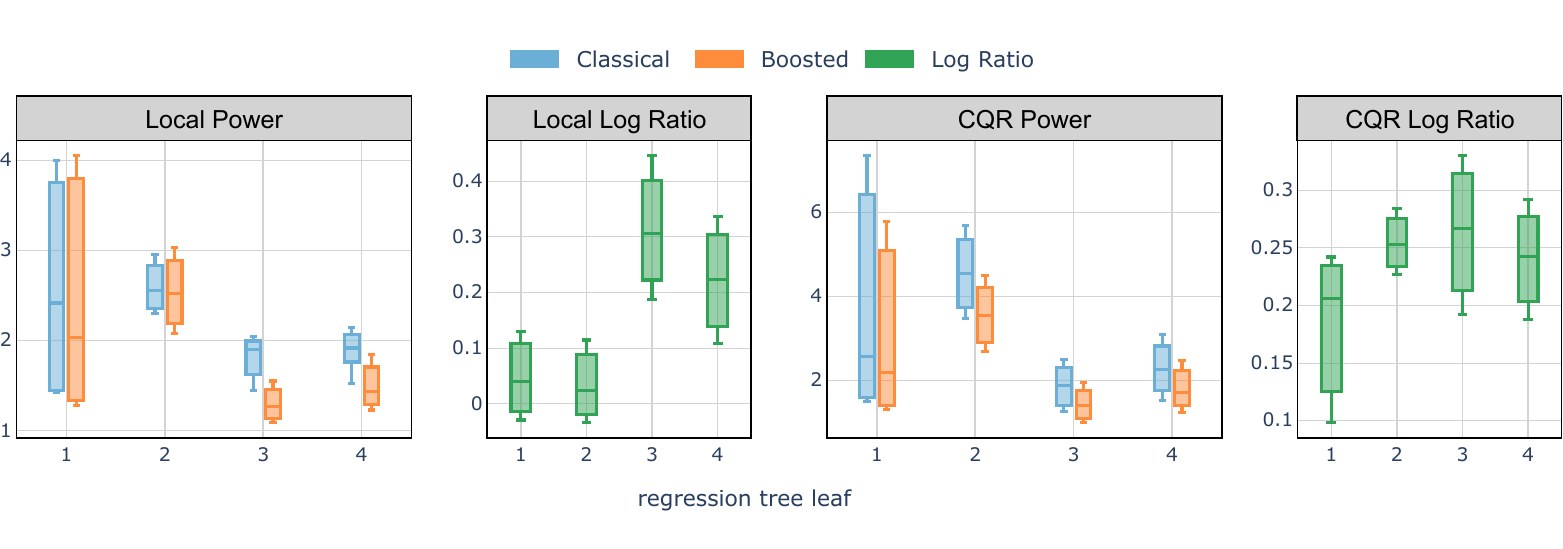}
         \caption{{meps-21} dateset. 
         }
     \end{subfigure}
      \hfill
     \begin{subfigure}[b]{\textwidth}
         \centering
         \includegraphics[width=\textwidth]{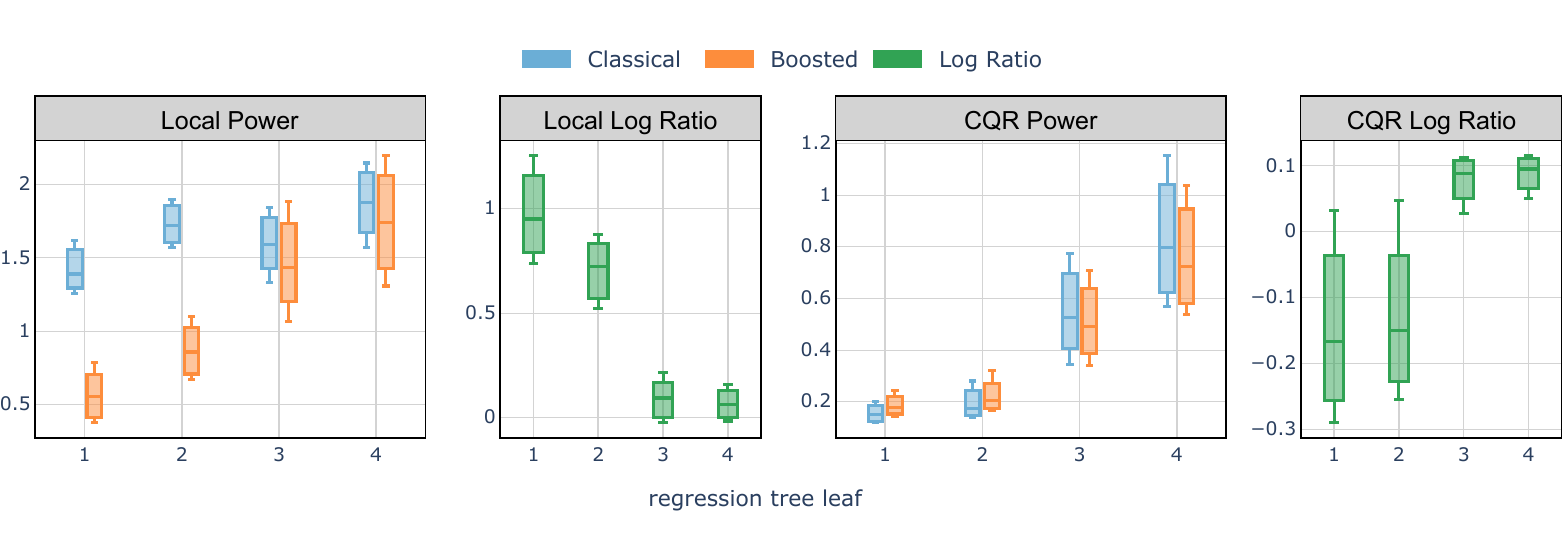}
         \caption{{bike} dateset. 
         }
     \end{subfigure}
    \caption{
    See the caption of Figure \ref{fig:len.real} for details.
    }
        \label{fig:supp.len.meps}

\vspace*{\fill}
\end{figure*}

\begin{figure*}[p] 
\centering
\vspace*{\fill}
 \begin{subfigure}[b]{\textwidth}
         \centering
\includegraphics[width=\textwidth]{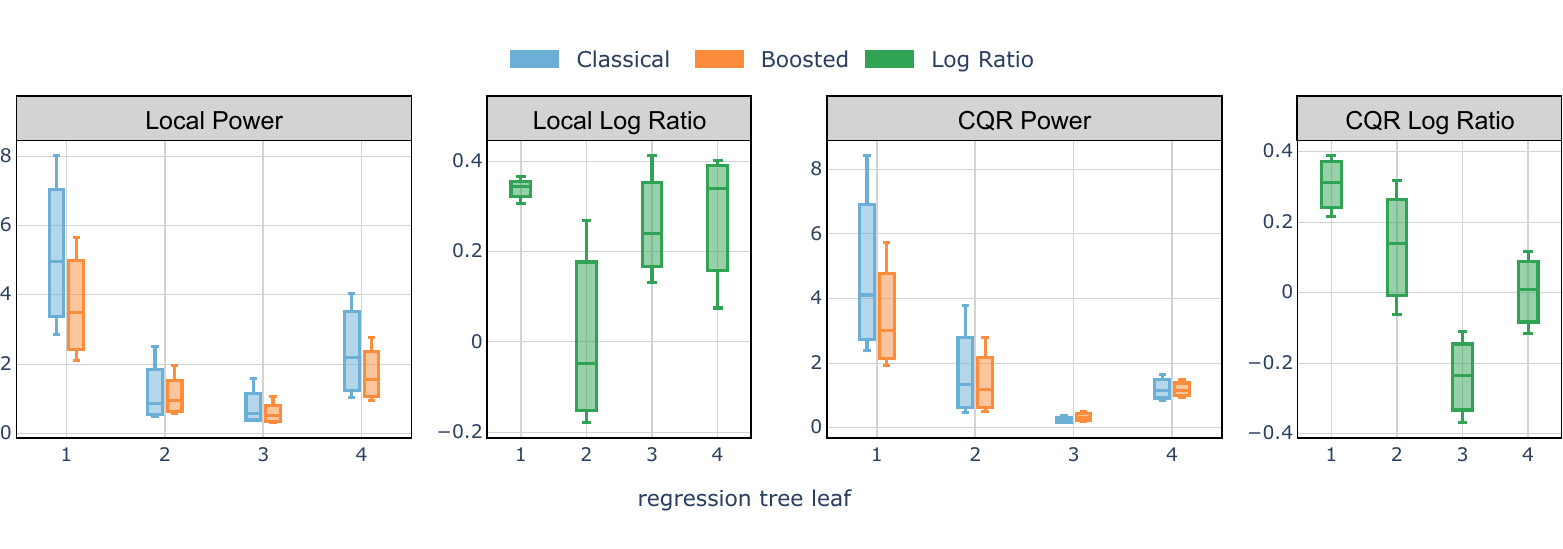}
         \caption{{facebook-1} dateset. 
         }
     \end{subfigure}
     \hfill
     \begin{subfigure}[b]{\textwidth}
         \centering
         \includegraphics[width=\textwidth]{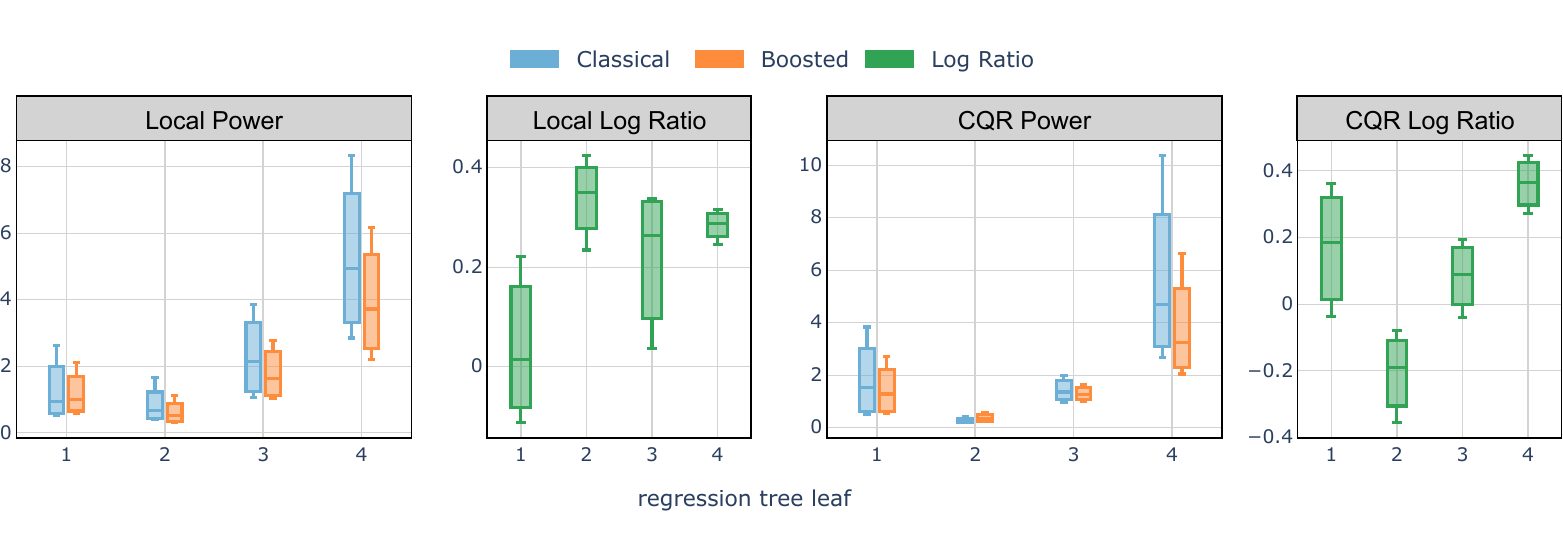}
         \caption{{facebook-2} dateset. 
         }
     \end{subfigure}
      \hfill
     \begin{subfigure}[b]{\textwidth}
         \centering
         \includegraphics[width=\textwidth]{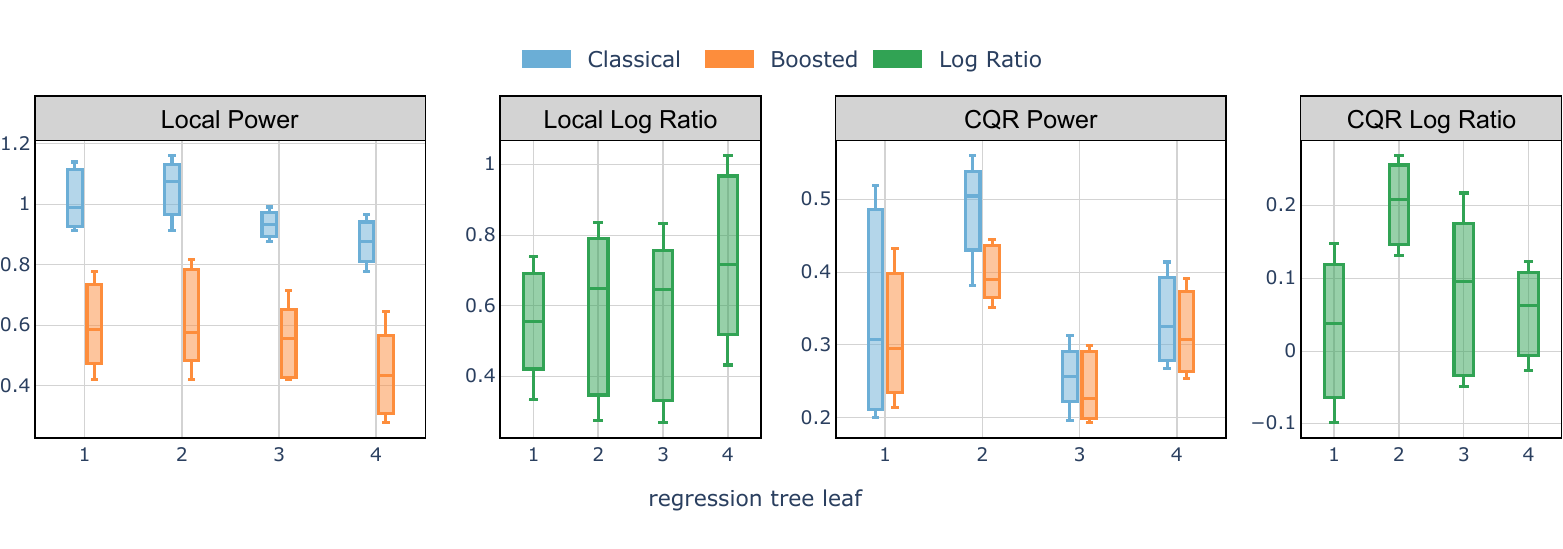}
         \caption{  {concrete} dateset. 
         }
     \end{subfigure}
    \caption{
    See the caption of Figure \ref{fig:len.real} for details.
    }
        \label{fig:supp.len.facebook}

\vspace*{\fill}
\end{figure*}

\begin{figure*}[p] 
\centering
\vspace*{\fill}
     \begin{subfigure}[b]{\textwidth}
         \centering
\includegraphics[width=\textwidth]{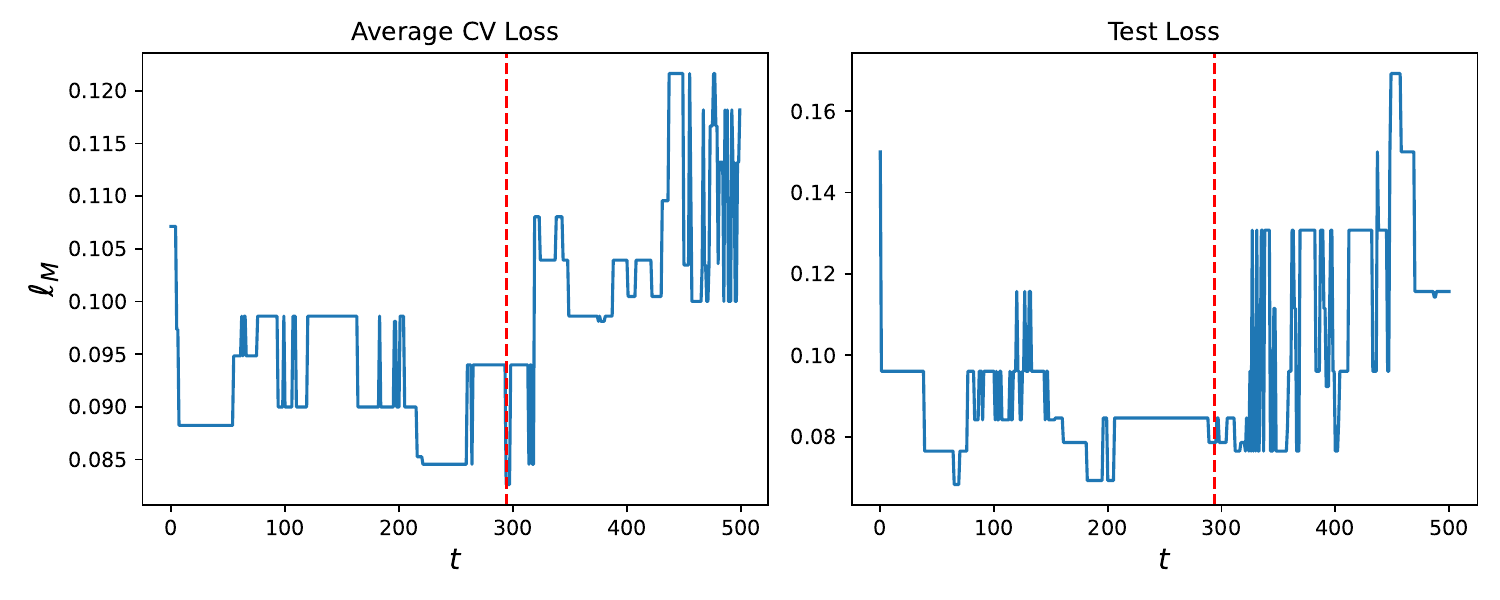}
         \caption{seed 9.}
     \end{subfigure}
     \hfill
      \begin{subfigure}[b]{\textwidth}
         \centering
\includegraphics[width=\textwidth]{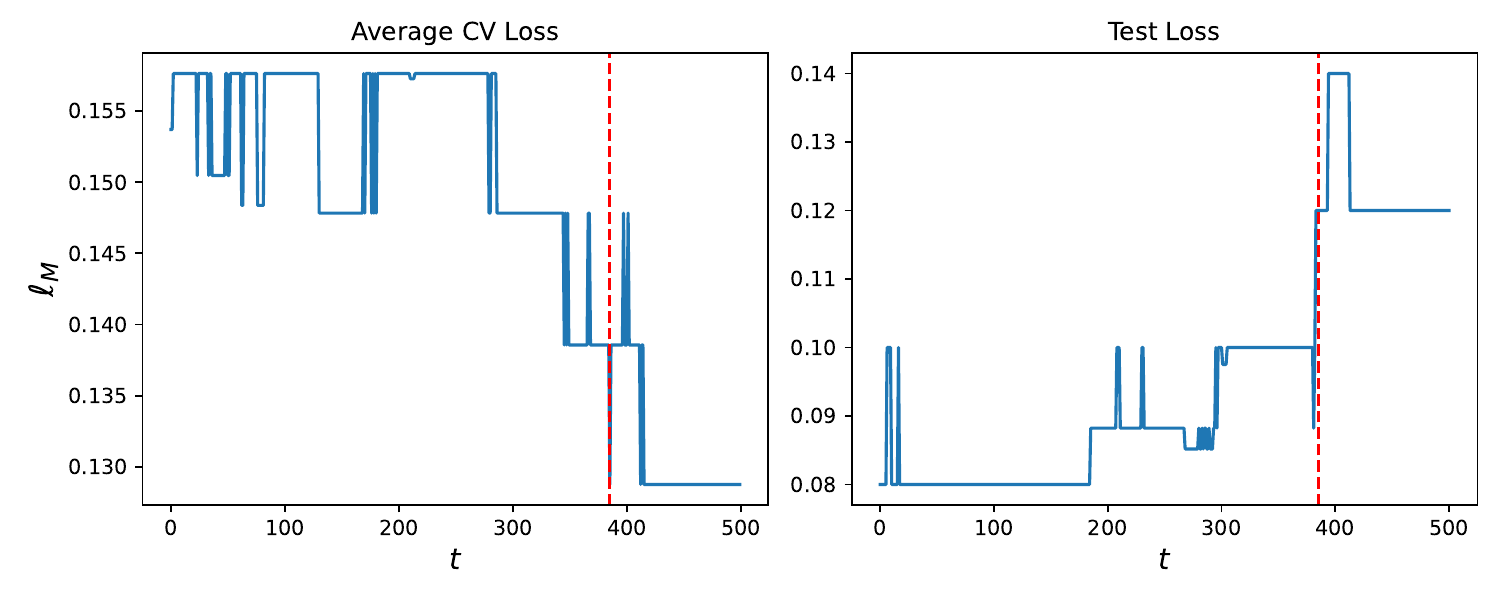}
         \caption{seed 8.}
     \end{subfigure}
     \hfill
      \begin{subfigure}[b]{\textwidth}
         \centering
\includegraphics[width=\textwidth]{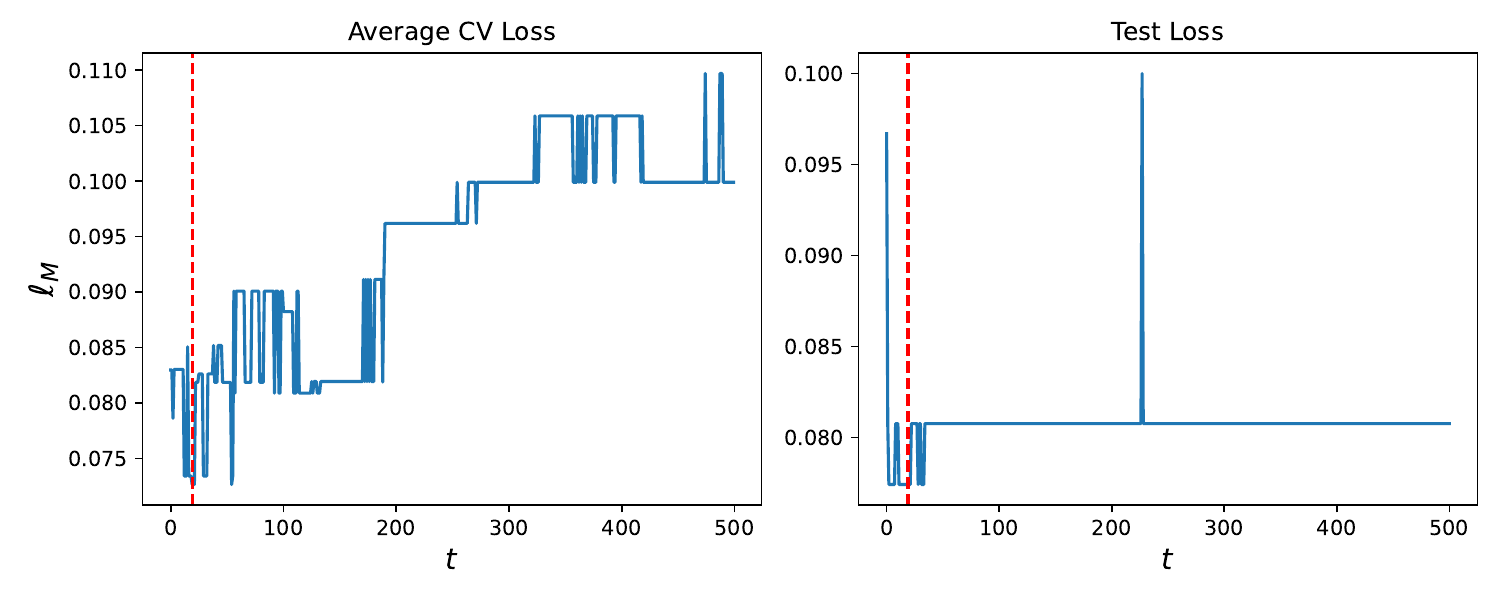}
         \caption{seed 7.}
     \end{subfigure}
    \caption{Empirical maximum deviation $\ell_M$ across $T=500$ boosting rounds evaluated on dataset {concrete} under random seeds 7, 8, 9, train-calibration ratio $60\%$ : The left panel illustrates the cross-validated loss, computed as the average across $k=3$ sub-calibration folds. The right panel displays the test loss. The optimal number of boosting rounds $\tau$, determined through cross-validation as specified in (\ref{eq:cv.step}), is highlighted in red.} 
    \label{fig:cv}
\vspace*{\fill}
\end{figure*}